\pdfoutput=1
\documentclass[lettersize,journal]{IEEEtran}
\usepackage{amsmath,amsfonts}
\usepackage{algorithmic}
\usepackage{algorithm}

\usepackage{array}
\usepackage[caption=false,font=footnotesize,labelfont=rm,textfont=rm]{subfig}
\usepackage{multirow}
\usepackage{textcomp}
\usepackage[hidelinks]{hyperref}
\usepackage{stfloats}
\usepackage{amssymb}
\usepackage{mathrsfs}
\usepackage{tabularx}
\usepackage{url}
\usepackage{booktabs}
\newtheorem{definition}{Definition} 
\newtheorem{example}{Example}
\newtheorem{theorem}{Theorem}
\newtheorem{lemma}{Lemma}
\newtheorem{corollary}{Corollary}
\usepackage{verbatim}
\usepackage{graphicx}
\usepackage{cite}
\usepackage{color}
\usepackage{tikz}
\usetikzlibrary{positioning}
\usetikzlibrary{arrows.meta, positioning, decorations.pathreplacing}
\usepackage{makecell}

\begin{document}
	
	\title{GRAND for Gaussian Intersymbol Interference Channels}
	
	\author{Zhuang Li, Wenyi Zhang,~\IEEEmembership{Senior Member,~IEEE}
		\thanks{Z. Li and W. Zhang are with Department of Electronic Engineering and Information Science, University of Science and Technology of China, Hefei, China (Email: wenyizha@ustc.edu.cn).}
	}
	
	
	
	\maketitle
	
	\begin{abstract}
		Channel decoding is challenging in communication channels exhibiting memory effects. In this work, we extend the recently proposed decoding paradigm of guessing random additive noise decoding (GRAND) to channels with memory, focusing on linear Gaussian intersymbol interference (ISI) channels. To characterize error patterns (EPs) under channel memory, we introduce the concept of error bursts and define sequence reliability to quantify the likelihood of an EP. Based on sequence reliability, we develop soft GRAND for ISI channels, termed SGRAND-ISI, which generalizes SGRAND to linear Gaussian ISI
		channels and is equivalent to maximum-likelihood (ML) decoding. We
		then develop ordered reliability bit GRAND (ORBGRAND) variants based
		on SGRAND-ISI that replace exact sequence reliabilities with rank-based
		approximations to reduce online computation. The proposed framework is first established for binary-input first-order ISI channels and subsequently extended to higher-order channel memory and higher-order modulation. Numerical experiments over first- and second-order ISI channels show that the proposed algorithms achieve multiple-dB gains over GRAND algorithms that ignore channel memory and can operate within $0.1$-$0.2$ dB of the ML benchmark. Compared with the recently proposed ORBGRAND-Approximate Independence algorithm, CDF-ORBGRAND-ISI provides a gain of at least $0.5$ dB at a block error rate (BLER) of $10^{-3}$. We further evaluate the online decoding complexity of the proposed
		GRAND-ISI algorithms using a module-wise complexity analysis.
	\end{abstract}
	
	
	\section{Introduction}
	\label{sec:intro}
	
	
	Reliable channel decoding over intersymbol interference (ISI) channels is important in a broad range of data-storage and communication systems. ISI commonly arises in magnetic recording channels \cite{siegel1991modulation}, high-speed electrical interconnects \cite{chae2024high,payne20056}, and multipath wireless channels \cite{proakis2001digital,rappaport2010wireless}. Since each received symbol is affected by multiple neighboring transmitted symbols, its reliability cannot generally be assessed in isolation. This motivates decoding methods that exploit the channel memory rather than relying solely on symbol-wise reliability information.

	Equalization has long been a standard approach for mitigating ISI. Maximum-likelihood sequence estimation (MLSE), commonly implemented using the Viterbi algorithm, provides optimal uncoded sequence detection for finite-memory channels \cite{forney1972maximum,viterbi1967error}, while lower-complexity techniques such as linear equalization (LE) \cite{lucky1968principles} and decision-feedback equalization (DFE) \cite{monsen1971feedback,belfiore1979decision} are widely used in practical systems. For channels with long impulse responses, channel shortening can first reduce the effective channel memory, followed by sequence detection over the shortened channel \cite{falconer1973adaptive,qureshi1973adaptive,lee1977maximum,beare1978choice}. 
	Nevertheless, equalization or channel shortening does not necessarily eliminate the channel memory completely, and the residual ISI may induce temporally correlated decoder inputs that contain useful information about the transmitted sequence. Conventional decoders designed under a memoryless-channel assumption may therefore suffer significant performance losses by ignoring the dependence induced by channel memory \cite{an2022keep,duffy2023using}. Turbo equalization addresses this issue by iteratively exchanging soft information between an equalizer and a channel decoder \cite{douillard1995iterative}, but requires an iterative receiver architecture and soft-input soft-output processing.
	
	Against this background, guessing random additive noise decoding (GRAND) provides a natural framework for incorporating channel statistics into the decoding process \cite{duffy2019capacity}. Rather than directly searching over the codebook, GRAND generates noise or error patterns (EPs) in decreasing order of likelihood and tests whether the corresponding corrected sequence is a codeword. This decoding principle separates the modeling of channel impairments from the algebraic structure of the code and is particularly attractive for short and medium-to-high-rate block codes. Most existing GRAND variants, however, are developed for memoryless channels \cite{duffy2021guessing,solomon2020soft,duffy2022ordered,liu2022orbgrand}.
	GRAND Markov Order (GRAND-MO) \cite{an2022keep} targets two-state Markov channels, avoids the commonly adopted interleaver, and shows notable gains over the Berlekamp-Massey decoder \cite{Berlekamp1968AlgebraicCT, massey1969shift}. However, GRAND-MO is a hard-decision decoder and thus cannot exploit soft reliability information. Ordered reliability bit GRAND with approximate independence (ORBGRAND-AI) was subsequently developed for soft decoding over correlated-noise channels, with Gauss-Markov noise serving as a representative model \cite{duffy2023using}. It partitions the received sequence into short blocks, accounts for the statistical dependence within each block, and approximates different blocks as independent. More recently, ORBGRAND-AI was investigated for equalized ISI channels, exploiting the noise correlation induced by front-end equalization through Gauss-Markov or autoregressive modeling \cite{duffy2025decoding}. These approaches extend GRAND beyond memoryless channels, but do not directly incorporate the original ISI channel structure into the ordering of EPs.
	
	To address this limitation, we develop a GRAND-based decoding framework that operates directly on the ISI channel model without interleaving. Unlike GRAND-MO, which orders EPs
	based on a Markov model of burst statistics, GRAND-ISI derives its EP
	ordering directly from the observation-dependent sequence metric of the
	original Gaussian ISI channel. We introduce an ISI-specific error-burst decomposition to characterize
	the structure of EPs induced by channel memory, and define the
	corresponding sequence reliability as the loss in the ISI sequence
	metric relative to the hard-detection sequence obtained by sequence
	detection \cite{viterbi1967error}. Based on this reliability measure, we generalize the GRAND decoding criterion in \cite{liu2022orbgrand} from memoryless channels to ISI channels. By querying EPs according to exact sequence reliabilities without a
	query limit, i.e., under complete querying, the resulting soft GRAND
	for ISI channels, termed SGRAND-ISI, achieves maximum-likelihood (ML)
	decoding.
	
	Although SGRAND-ISI achieves ML decoding, it requires the exact
	sequence reliability values during the decoding process, resulting in
	substantial online computation. To reduce this processing burden, we
	extend ORBGRAND to ISI channels and term the resulting decoder
	ORBGRAND-ISI. ORBGRAND-ISI replaces the exact sequence reliability values with the corresponding ranks, thereby preserving the reliability-ordered querying principle while substantially reducing real-time computation. Inspired by the cumulative distribution function (CDF)-based rank companding principle recently developed for ORBGRAND \cite{duffy2022ordered,wan2024approaching,li2026orbgrand}, we further propose CDF-ORBGRAND-ISI, which applies the inverse CDF of sequence reliability to the normalized reliability ranks. The resulting companded values provide a more accurate approximation to the exact sequence reliability values and yield a notable performance improvement over ORBGRAND-ISI with modest additional computation.
	
	We conduct numerical experiments to assess the proposed decoders over first- and second-order ISI channels. Simulation results show that, for a BCH code, CDF-ORBGRAND-ISI achieves a gain of at least $2$ dB over ORBGRAND, which ignores channel memory, even at a block error rate (BLER) of $10^{-1}$. Moreover, compared with ORBGRAND-AI, CDF-ORBGRAND-ISI provides a gain of at least $0.5$ dB at $\mathrm{BLER}=10^{-3}$ and operates within $0.1$-$0.2$ dB of the ML benchmark. We also evaluate the decoding complexity of the proposed GRAND-ISI algorithms using a module-wise complexity analysis.

	The remainder of this paper is organized as follows. Section \ref{sec:Model} introduces the system model and reviews the GRAND framework. Section \ref{sec:First} characterizes EPs in first-order ISI channels and develops the corresponding GRAND-ISI algorithms. Section \ref{sec:higher-order-ISI} extends the framework to higher-order ISI channels, while Section \ref{sec:higher-order-modulation} considers higher-order modulation. Section \ref{sec:numerical-results} presents numerical results, and Section \ref{Conclusion} concludes the paper.

	\section{Preliminaries}\label{sec:Model}
	
	Throughout the paper, symbols in sans-serif font, such as
	$\mathsf{Y}$, denote random variables, while the corresponding ordinary
	symbols, such as $y$, denote their realizations. Underlined
	symbols denote vectors; for example, $\underline{\mathsf{Y}}$ and
	$\underline{y}$ denote a random vector and one of its realizations,
	respectively. A symbol without an underline denotes a scalar, and
	$y_i$ denotes the $i$-th element of $\underline{y}$. For integers
	$a\leq b$, we use
	$[a:b]\triangleq\{a,a+1,\ldots,b\}$. For a vector
	$\underline{u}=(u_1,\ldots,u_n)$, its support is denoted by $\operatorname{supp}(\underline{u})\triangleq\{i:u_i\neq 0\}$.
	A summary of the main symbols used in the system model is provided
	in Table~\ref{tab:notation}.
	\begin{table}[htbp]
		\centering
		\caption{Main symbols used in the system model.}
		\label{tab:notation}
		\begin{tabular}{c|l}
			\hline
			Symbol & Meaning \\
			\hline
			$w$
			& transmitted message \\
			
			$\underline{c}(w)$
			& binary codeword corresponding to message $w$ \\
			
			$\underline{x}(w)$
			& modulated channel-input sequence corresponding to
			$\underline{c}(w)$ \\
			
			$\underline{\mathsf{Y}}$
			& channel-output random vector \\
			
			$\mathsf{Z}_i$
			& Gaussian noise at the $i$-th channel use \\
			
			$h_l$
			& $l$-th channel impulse-response coefficient \\
			
			$L$
			& ISI channel memory order \\
			\hline
		\end{tabular}
	\end{table}

	\subsection{System Model}\label{Channel Model}
	
	We consider binary phase-shift keying (BPSK) modulation over an ISI channel in which the input-output relationship at the $i$-th channel use is
	\begin{align}\label{eqn:ISI-channel}
		\mathsf{Y}_i=\sum_{l=0}^{L}h_l \mathsf{X}_{i-l}+\mathsf{Z}_i,
	\end{align}
	where $\mathsf{X}_i\in\{+1,-1\}$ denotes the channel input,
	$\{\mathsf{Z}_i\}$ is an independent and identically distributed
	(i.i.d.) Gaussian noise sequence with
	$\mathsf{Z}_i\sim\mathcal{N}(0,\sigma^2)$, and $\mathsf{Y}_i$ denotes
	the channel output. The coefficients $h_0,\ldots,h_L$ constitute the
	channel impulse response and are assumed to be perfectly known at the
	receiver. Without loss of generality, we assume that $\sum_{l=0}^{L}h_l^2=1$, so the signal-to-noise ratio (SNR) is ${1}/{\sigma^2}$.

	We adopt a codebook $\mathcal{C}$ with code length $N$ and code rate $R$ bits per channel use. The number of messages is given by $M = \lceil 2^{NR} \rceil$. When message $w \in \left\{1, \ldots, M\right\}$ is transmitted, the corresponding codeword is denoted by $\underline{c}(w) = [c_1(w), \ldots, c_N(w)]$, whose elements are drawn from $\left\{0, 1\right\}$. Under BPSK modulation, $\underline{c}(w)$ is mapped to the channel-input sequence $\underline{x}(w)=[x_1(w),\ldots,x_N(w)]\in\{+1,-1\}^{N}$,
	where
	\begin{align}
		x_i(w)=1-2c_i(w), \qquad i=1,\ldots,N.
	\end{align}
	
	For an index set
	$\mathcal{S}\subseteq\{1,\ldots,N\}$ and a BPSK sequence
	$\underline{x}\in\{+1,-1\}^{N}$, let
	$\mathcal{F}_{\mathcal{S}}(\underline{x})$ denote the sequence
	obtained by reversing the signs of the symbols indexed by
	$\mathcal{S}$, i.e.,
	\begin{align}\label{eq:flipping-operator}
		\left[
		\mathcal{F}_{\mathcal{S}}(\underline{x})
		\right]_i
		=
		\begin{cases}
			-x_i, & i\in\mathcal{S},\\
			x_i, & i\notin\mathcal{S}.
		\end{cases}
	\end{align}

	Upon receiving a realization $\underline{y}$ of the channel output sequence $\underline{\mathsf{Y}}$, for any channel-input sequence
	$\underline{x}\in\{+1,-1\}^{N}$, not necessarily corresponding to a
	codeword in $\mathcal{C}$, we define the sequence metric as
	\begin{align}\label{eq:sequence-metric}
		\Gamma(\underline{x}, \underline{y})\triangleq\sum_{i=1}^{N} \ln p_{\mathsf{Y}_i\vert \underline{\mathsf{X}}_{i-L:i}}(y_i\vert \underline{x}_{i-L:i}).
	\end{align}
	where $\underline{x}_{i-L:i}=[x_{i-L},\ldots,x_i]$. The ML decoding rule is then given by
	\begin{align}\label{eq:ML}
		\widehat{w}_{\mathrm{ML}}
		&=
		\arg\max_{w=1,\ldots,M}
		p_{\underline{\mathsf{Y}}|\underline{\mathsf{X}}}
		\left(
		\underline{y}|\underline{x}(w)
		\right)
		\nonumber\\
		&=
		\arg\max_{w=1,\ldots,M}
		\Gamma\bigl(\underline{x}(w),\underline{y}\bigr).
	\end{align}
	
	In addition, we define the hard-detection sequence as the unconstrained ML sequence estimate
	\begin{align}\label{Hard}
		\underline{x}^{*}
		\triangleq
		\arg\max_{\underline{x}\in\{+1,-1\}^{N}}
		\Gamma(\underline{x},\underline{y}),
	\end{align}
	which can be efficiently obtained using the Viterbi algorithm. The
	corresponding binary hard-decision sequence is denoted by
	$\underline{c}^{*}=[c_1^{*},\ldots,c_N^{*}]$, where $c_i^{*}=(1-x_i^{*})/2$ for $i=1,\ldots,N$.

	\subsection{Guessing Random Additive Noise Decoding}\label{introduction of GRAND}
	
	GRAND is a noise-centric or, more precisely, an EP-centric decoding paradigm \cite{duffy2019capacity}. Rather than directly searching over the codebook, GRAND sequentially generates EPs and tests whether removing each EP from the binary hard-decision sequence produces a codeword. When the EPs are queried in
	decreasing order of likelihood, GRAND implements ML decoding. This principle applies not only to memoryless channels but also to channels with memory, provided that the corresponding EP statistics are properly characterized. 
	
	Early GRAND constructions mainly focused on hard-decision
	binary channels. The subsequent incorporation of soft reliability
	information has led to several soft-input variants. Symbol reliability GRAND \cite{duffy2021guessing} augments each hard decision with a one-bit reliability indicator and uses this information to refine the EP ordering
	\cite{duffy2021guessing}. SGRAND \cite{solomon2020soft} generates EPs based on the magnitudes of log-likelihood ratios (LLRs), thereby fully exploiting soft reliability information and achieving ML decoding.  However, generating EPs from exact LLR values requires substantial online computation and is less amenable to efficient hardware implementation. ORBGRAND \cite{duffy2022ordered} utilizes the relationship of ranking among LLR magnitudes rather than their exact values to generate EPs. This approach enables efficient generation of EPs and hence facilitates hardware implementation \cite{condo2021high, abbas2022high, condo2022fixed}. Notably, ORBGRAND has been shown to achieve an information rate that is very close to channel capacity for commonly used channel models including the additive white Gaussian noise (AWGN) channel \cite{liu2022orbgrand, li2024orbgrand}. Moreover, rank companding ensures that ORBGRAND is exactly capacity-achieving, and the resulting variant is termed CDF-ORBGRAND \cite{li2026orbgrand}. 
	
	We next describe the generic operation of GRAND under the system model
	in Section~\ref{Channel Model}. Given the channel output sequence $\underline{y}$, the receiver first obtains the unconstrained hard-detection sequence $\underline{x}^{*}$ from \eqref{Hard} and maps it to the corresponding binary sequence $\underline{c}^{*}=[c_1^{*},\ldots,c_N^{*}]$. GRAND then generates an
	ordered list of binary EPs $\underline{e}^{(1)},\underline{e}^{(2)},\ldots
	\in\{0,1\}^{N}$.
	For each  $\underline{e}^{(q)}$, it tests whether
	$\underline{c}^{*}\oplus\underline{e}^{(q)}$ satisfies the code
	constraints. For the binary linear codes considered in this work,
	$\mathbf{H}$ denotes the parity-check matrix, and codeword membership
	is verified through the standard zero-syndrome test
	$\mathbf{H}(\underline{c}^{*}\oplus\underline{e}^{(q)})^{\mathsf T}
	=\underline{0}$. The procedure terminates when a codeword is found
	or when the prescribed query limit is reached. 
	
%
%
%
%
%
%
%
	
	The ordering of the EPs is central to the performance of GRAND. For a hard-decision memoryless channel, the EPs can be ordered according to
	their Hamming weights \cite{duffy2021guessing}. For a soft-output
	memoryless channel, the optimal ordering can be constructed from sums
	of bitwise LLR magnitudes \cite{solomon2020soft,li2024orbgrand}. In an ISI channel, however, channel memory couples neighboring symbols,
	so the likelihood of an EP generally cannot be determined from
	individual bit reliabilities alone. An optimal or near-optimal EP
	ordering must therefore account for the sequence dependence induced by
	the channel memory, which is the focus of the following sections.

	\section{GRAND-ISI for First-order ISI channels}\label{sec:First}
	
	In this section, we investigate the structure and statistical properties
	of EPs in first-order ISI channels, i.e., $L=1$ in
	\eqref{eqn:ISI-channel}, and develop a GRAND framework tailored to such
	channels, termed GRAND-ISI, together with several practical variants.
	
	\subsection{Error Bursts and Sequence Reliability}

	When message $w$ is transmitted, the corresponding EP relative to the binary hard-decision sequence $\underline{c}^{*}$ is $\underline{e}(w)=\underline{c}^{*}\oplus\underline{c}(w)$. By the BPSK mapping $x_i=1-2c_i$, the transmitted channel-input
	sequence can equivalently be expressed as
	\begin{align}\label{eq:EP-flipping-relation}
		\underline{x}(w)
		=
		\mathcal{F}_{\operatorname{supp}(\underline{e}(w))}
		\bigl(\underline{x}^{*}\bigr).
	\end{align}

	\begin{definition}[Error burst decomposition]
		\label{defn:error-burst}
		For a nonzero EP $\underline{e}\in\{0,1\}^{N}$, its support admits a
		unique partition into maximal contiguous intervals,
		\begin{align}
			\operatorname{supp}(\underline{e})
			=
			\bigcup_{k=1}^{\zeta_1(\underline{e})}
			\mathcal{B}^{(1)}_k(\underline{e}),
		\end{align}
		where
		$\mathcal{B}^{(1)}_k(\underline{e})=[s_k:t_k]$.
		The intervals are indexed such that
		\begin{align}
			s_1<s_2<\cdots<s_{\zeta_1(\underline{e})},
		\end{align}
		and consecutive intervals satisfy
		\begin{align}
			t_k+1<s_{k+1},
			\qquad
			k=1,\ldots,\zeta_1(\underline{e})-1.
		\end{align}
		Each $\mathcal{B}^{(1)}_k(\underline{e})$ is called an error burst of
		$\underline{e}$, and $\zeta_1(\underline{e})$ denotes the number of
		error bursts. For $\underline{e}=\underline{0}$, the decomposition is
		empty and $\zeta_1(\underline{0})=0$.
	\end{definition}

	The following example illustrates how the support of an EP is
	partitioned into its first-order error bursts.
	\begin{example}
		Consider the EP $\underline{e}=[1,1,0,1,0,0,1,1]$. Its support is $\operatorname{supp}(\underline{e})=\{1,2,4,7,8\}$. According to Definition~\ref{defn:error-burst}, its error burst
		decomposition is
		\begin{align}
			\mathcal{B}^{(1)}_1(\underline{e})=[1:2],\, \mathcal{B}^{(1)}_2(\underline{e})=[4:4], \,\mathcal{B}^{(1)}_3(\underline{e})=[7:8],
		\end{align}
		and hence $\zeta_1(\underline{e})=3$.
	\end{example}

	For a length-$N$ EP, an error burst of size $\ell$ can occupy any of
	the $N-\ell+1$ intervals
	$[s:s+\ell-1]$, where $s=1,\ldots,N-\ell+1$. Therefore, the total
	number of possible error bursts is
	\begin{align}
		V_1
		\triangleq
		\sum_{\ell=1}^{N}(N-\ell+1)
		=
		\frac{N(N+1)}{2}.
	\end{align}

	\begin{definition}[Sequence reliability]
		\label{defn:sequence-reliability}
		For any index set
		$\mathcal{S}\subseteq\{1,\ldots,N\}$, its sequence reliability is
		defined as
		\begin{align}\label{eq:sequence-reliability}
			\operatorname{Rel}(\mathcal{S})
			\triangleq
			\Gamma(\underline{x}^{*},\underline{y})
			-
			\Gamma\left(
			\mathcal{F}_{\mathcal{S}}(\underline{x}^{*}),
			\underline{y}
			\right).
		\end{align}
	\end{definition}
	The sequence reliability quantifies the loss in the sequence metric
	caused by reversing the symbols indexed by $\mathcal{S}$. Since
	$\underline{x}^{*}$ maximizes the sequence metric,
	$\operatorname{Rel}(\mathcal{S})\geq 0$. Among the flipped sequences
	under consideration, a smaller sequence reliability indicates a
	smaller metric loss and, equivalently, a larger conditional likelihood
	$p_{\underline{\mathsf{Y}}|\underline{\mathsf{X}}}
	\bigl(\underline{y}|\mathcal{F}_{\mathcal{S}}(\underline{x}^{*})\bigr)$.

	\subsection{General Form of GRAND-ISI}
	
	A general GRAND-ISI decoder is characterized by a nonnegative function
	$\gamma(\underline{y},\mathcal{S})$, which assigns a reliability
	metric to an index set $\mathcal{S}\subseteq\{1,\ldots,N\}$. For an EP
	$\underline{e}\in\{0,1\}^{N}$, let
	$\{\mathcal{B}^{(1)}_k(\underline{e})\}_{k=1}^{\zeta_1(\underline{e})}$
	be its unique error burst decomposition specified in
	Definition~\ref{defn:error-burst}. Define
	\begin{align}\label{eq:aggregate-reliability}
		\rho_{\gamma}(\underline{y},\underline{e})
		\triangleq
		\sum_{k=1}^{\zeta_1(\underline{e})}
		\gamma\left(
		\underline{y},
		\mathcal{B}^{(1)}_k(\underline{e})
		\right),
	\end{align}
	where the empty sum associated with $\underline{e}=\underline{0}$ is defined as zero.
	
	The general operation of GRAND-ISI can then be described as follows:
	\begin{itemize}
		\item Sort all EPs $\underline{e}\in\{0,1\}^{N}$ in ascending order of
		$\rho_{\gamma}(\underline{y},\underline{e})$, where ties may be broken arbitrarily.
		
		\item Following the resulting order, apply each EP $\underline{e}$ to the binary hard-decision sequence to obtain $\underline{\widetilde{c}}=\underline{c}^{*}\oplus\underline{e}$.

		\item Test whether $\underline{\widetilde{c}}$ is a codeword. If so,
		terminate and declare it as the decoding result; otherwise, proceed
		to the next EP until a codeword is found or the prescribed maximum
		number of queries $Q$ is reached.
	\end{itemize}
	
	When the abandonment threshold is set to
	$Q=2^N$, the query order contains all possible EPs. Since each codeword
	$\underline{c}(w)$ corresponds to the unique EP
	$\underline{e}(w)=\underline{c}^{*}\oplus\underline{c}(w)$, the
	decoder can be equivalently expressed as
	\begin{align}\label{eq:general-GRAND-ISI-decoder}
		\widehat{w}_{\gamma}
		&=
		\arg\min_{w=1,\ldots,M}
		\rho_{\gamma}\bigl(
		\underline{y},\underline{e}(w)
		\bigr)
		\nonumber\\
		&=
		\arg\min_{w=1,\ldots,M}
		\sum_{k=1}^{\zeta_1(\underline{e}(w))}
		\gamma\left(
		\underline{y},
		\mathcal{B}^{(1)}_k(\underline{e}(w))
		\right).
	\end{align}
	
	Different choices of $\gamma(\underline{y},\mathcal{S})$ induce
	different EP orderings and hence different members of the GRAND-ISI
	family. This framework naturally extends the unified representation of
	GRAND algorithms developed in~\cite{liu2022orbgrand} to first-order
	ISI channels. Some representative choices of $\gamma$ are listed
	below.
	\begin{itemize}
		\item When $\gamma(\underline{y},\mathcal{S})=\operatorname{Rel}(\mathcal{S})$,
		where $\operatorname{Rel}(\mathcal{S})$ is defined in
		Definition~\ref{defn:sequence-reliability}, the resulting decoder
		is termed SGRAND-ISI, which extends
		SGRAND~\cite{solomon2020soft} to first-order ISI channels.
		See Section~\ref{subsec:sgrand-isi} for details.
		
		\item When $\gamma(\underline{y},\mathcal{S})=r(\mathcal{S})$,
		where $r(\mathcal{S})\in\{1,\ldots,V_1\}$ denotes the rank of
		$\operatorname{Rel}(\mathcal{S})$ among the sequence reliabilities
		of all $V_1$ possible error bursts, from $1$ (the smallest) to
		$V_1$ (the largest). The resulting decoder is termed ORBGRAND-ISI,
		which extends ORBGRAND~\cite{duffy2022ordered} to first-order ISI
		channels.
		
		\item When $\gamma(\underline{y},\mathcal{S})=\Psi^{-1}\left(\frac{r(\mathcal{S})}{V_1+1}\right)$, where $\Psi(\cdot)$ denotes the CDF of the sequence reliability of an error burst selected uniformly from the $V_1$ possible error bursts under the considered channel model, the resulting decoder is termed CDF-ORBGRAND-ISI. The inverse-CDF rank-companding idea underlying CDF-ORBGRAND can be traced back to~\cite{duffy2022ordered}, while its exact capacity-achieving property over binary-input memoryless channels was established in~\cite{li2026orbgrand}.
	\end{itemize}
	
	\subsection{Properties of Sequence Reliability}\label{subsec:properties-1st-ISI}
	
	In this subsection, we establish several structural properties of
	sequence reliability that facilitate the efficient computation of
	EP reliability metrics for GRAND-ISI algorithms over first-order ISI
	channels.
	
	\begin{lemma}	
		\label{lemma:1}
		A nonempty index set $\mathcal{S}\subseteq\{1,\ldots,N\}$ is called
		decomposable if it can be written as
		$\mathcal{S}=\mathcal{S}_1\cup\mathcal{S}_2$, where
		$\mathcal{S}_1$ and $\mathcal{S}_2$ are nonempty and satisfy
		\begin{align}
			\mathcal{S}_1\cap\mathcal{S}_2=\emptyset, 	\quad\min(\mathcal{S}_2)-\max(\mathcal{S}_1)>1.
		\end{align}
		Then, we can obtain
		\begin{align}\label{eq:decomposable-additivity}
			\operatorname{Rel}(\mathcal{S})
			=
			\operatorname{Rel}(\mathcal{S}_1)
			+
			\operatorname{Rel}(\mathcal{S}_2).
		\end{align}
	\end{lemma}
	
	\begin{IEEEproof}
		See Appendix~\ref{Lemma1}.
	\end{IEEEproof}
	
	By repeatedly applying Lemma~\ref{lemma:1} to the error burst
	decomposition in Definition~\ref{defn:error-burst}, we obtain, for
	any EP $\underline{e}\in\{0,1\}^{N}$,
	\begin{align}\label{eq:error-burst-additivity}
		\operatorname{Rel}
		\bigl(\operatorname{supp}(\underline{e})\bigr)
		=
		\sum_{k=1}^{\zeta_1(\underline{e})}
		\operatorname{Rel}\left(
		\mathcal{B}^{(1)}_k(\underline{e})
		\right).
	\end{align}
	
	A nonempty index set is called non-decomposable if it is a contiguous interval.
	
	\begin{lemma}
		\label{lemma:2}
		For any non-decomposable index set $[a:b]$, where
		$1\leq a\leq b\leq N$,
		\begin{align}\label{eq:reliability-contiguous-set}
			\operatorname{Rel}([a:b])
			=
			\sum_{i=a}^{b}\operatorname{Rel}(\{i\})
			+
			\frac{4h_0h_1}{\sigma^2}
			\sum_{i=a}^{b-1}x_i^*x_{i+1}^*,
		\end{align}
		where the second sum is understood as zero when $a=b$.
	\end{lemma}
	
	\begin{IEEEproof}
		See Appendix~\ref{Lemma2}.
	\end{IEEEproof}
	
	Together with \eqref{eq:error-burst-additivity},
	Lemma~\ref{lemma:2} enables
	$\operatorname{Rel}(\operatorname{supp}(\underline{e}))$
	to be efficiently evaluated from the error burst decomposition of
	$\underline{e}$. In particular, once the sequence reliabilities of
	all error bursts of size one are available, those of all
	$V_1=N(N+1)/2$ possible error bursts can be recursively evaluated.
	The following example illustrates this procedure.

	
	\begin{example}
		Consider a first-order ISI channel with $N=4$,
		$4h_0h_1/\sigma^2=1.2$, and
		$\underline{x}^{*}=[1,1,1,-1]$. Suppose that the size-one burst
		reliabilities are $1.14$, $0.95$, $0.39$, and $3.96$ for positions
		$1,\ldots,4$, respectively. By Lemma~\ref{lemma:2},
		$\operatorname{Rel}(\{1,2\})=1.14+0.95+1.2=3.29$.
		The remaining burst reliabilities are obtained similarly, so all
		$V_1=N(N+1)/2=10$ values follow from those of size one.
	\end{example}

	\subsection{SGRAND-ISI}\label{subsec:sgrand-isi}

	Similar to the fact that for memoryless channels SGRAND is equivalent to ML decoding, we have the following result regarding SGRAND-ISI.
	\begin{theorem}\label{thm:SGRAND-ISI-ML}
		Under the first-order ISI model, if SGRAND-ISI is allowed to query
		all EPs, i.e., $Q=2^N$, it is equivalent to ML decoding.
	\end{theorem}
	\begin{IEEEproof}
		See Appendix~\ref{Equivalence}.
	\end{IEEEproof}
	
	We next describe a realization of SGRAND-ISI based on
	error burst combinations. Index the $V_1$ possible error bursts as
	$\mathcal{B}_1,\ldots,\mathcal{B}_{V_1}$ in increasing
	order of their sequence reliabilities, such that $\operatorname{Rel}\left(\mathcal{B}_1\right)\leq\cdots\leq\operatorname{Rel}\left(\mathcal{B}_{V_1}\right)$, with ties broken arbitrarily.
	
	A burst-selection vector
	$\underline{\alpha}
	=[\alpha_1,\ldots,\alpha_{V_1}]
	\in\{0,1\}^{V_1}$
	specifies a combination of error bursts, where $\alpha_v=1$
	indicates that $\mathcal{B}_v$ is selected. Not every
	burst-selection vector corresponds to a valid error burst
	decomposition, since the selected bursts may overlap or be adjacent. By Definition~\ref{defn:error-burst}, distinct error
	bursts are disjoint maximal contiguous intervals and must therefore be
	separated by at least one zero position. Hence, overlapping bursts
	cannot appear simultaneously in the same decomposition, while adjacent
	bursts would instead form a single larger error burst.
	A burst-selection vector $\underline{\alpha}$ is called valid if any
	two selected bursts are separated by at least one unselected
	position, i.e., for any distinct $u$ and $v$ with
	$\alpha_u=\alpha_v=1$,
	\begin{align}\label{eq:first-order-validity}
		\max\left(\mathcal{B}_u\right)+1
		<
		\min\left(\mathcal{B}_v\right),
	\end{align}
	after interchanging $u$ and $v$ if necessary. The all-zero vector is
	also regarded as valid.
	
	Each valid burst-selection vector $\underline{\alpha}$ uniquely
	induces an EP $\underline{e}(\underline{\alpha})$ through $\operatorname{supp}\left(\underline{e}(\underline{\alpha})\right)=\bigcup_{v:\alpha_v=1}\mathcal{B}_v$.
	Since the selected bursts are nonoverlapping and nonadjacent, they
	constitute the unique error burst decomposition of
	$\underline{e}(\underline{\alpha})$. Conversely, the unique
	error burst decomposition of every EP determines a unique valid
	burst-selection vector. Hence, there is a one-to-one correspondence
	between the valid burst-selection vectors and the binary EPs in
	$\{0,1\}^{N}$.
	
	For SGRAND-ISI, 
	$\gamma(\underline{y},\mathcal{S})=\operatorname{Rel}(\mathcal{S})$.
	Hence, for any valid burst-selection vector
	$\underline{\alpha}$, \eqref{eq:aggregate-reliability} and
	\eqref{eq:error-burst-additivity} yield
	\begin{align}\label{eq:sgrand-isi-selection-metric}
		\rho_{\gamma}
		\left(
		\underline{y},
		\underline{e}(\underline{\alpha})
		\right)
		&=
		\sum_{v=1}^{V_1}
		\alpha_v
		\operatorname{Rel}
		\left(
		\mathcal{B}_v
		\right)
		=
		\operatorname{Rel}
		\left(
		\operatorname{supp}
		\left(
		\underline{e}(\underline{\alpha})
		\right)
		\right).
	\end{align}

	SGRAND-ISI orders the valid burst-selection vectors
	in ascending order of the metric
	in \eqref{eq:sgrand-isi-selection-metric}, with ties broken arbitrarily.
	For each $\underline{\alpha}$ in this order, the corresponding EP is
	applied to the binary hard-decision sequence to form $\underline{\widetilde{c}}=\underline{c}^{*}\oplus\underline{e}(\underline{\alpha})$. If $\underline{\widetilde{c}}$ is a codeword, the decoder terminates and outputs it; otherwise, it proceeds to the next valid
	burst-selection vector until a codeword is found or the abandonment
	threshold $Q$ is reached.
	
	To avoid explicitly enumerating and sorting all valid burst-selection
	vectors, SGRAND-ISI adopts the online generation procedure in
	\cite[Algorithm~2]{solomon2020soft}, with the possible error bursts
	and their sequence reliabilities treated as the  events and
	corresponding weights, respectively. Burst-selection vectors are thereby generated in nondecreasing order of the sums of the sequence reliabilities of their selected error bursts. Compared with memoryless SGRAND, each generated vector is first checked against the validity condition in \eqref{eq:first-order-validity}. Each valid vector is then mapped to its corresponding EP, which is applied to the binary hard-decision sequence to form a resulting sequence. This sequence is subsequently tested using the parity-check matrix, whereas invalid vectors are discarded. Since this validity check only removes invalid burst combinations, it
	preserves the relative order of the valid vectors. By \eqref{eq:error-burst-additivity}, the corresponding EPs are therefore queried in nondecreasing order of their exact sequence reliabilities, as required by SGRAND-ISI. The resulting procedure is summarized in
	Algorithm~\ref{alg:SGRAND-ISI}.
	
	\begin{algorithm}[t]
		\caption{SGRAND-ISI for First-Order ISI Channels}
		\label{alg:SGRAND-ISI}
		\begin{algorithmic}[1]
			\renewcommand{\algorithmicrequire}{\textbf{Input:}}
			\renewcommand{\algorithmicensure}{\textbf{Output:}}
			\REQUIRE $\underline{c}^{*}$, $\mathbf{H}$,
			$\{\mathcal{B}^{(1)}_v\}_{v=1}^{V_1}$,
			$\{\operatorname{Rel}(\mathcal{B}^{(1)}_v)\}_{v=1}^{V_1}$,
			$Q$
			\ENSURE $\underline{\widehat{c}}$ or \textsc{Abandon}
			
			\STATE $q\leftarrow 0$
			\WHILE{$q<Q$}
			\STATE Generate the next $\underline{\alpha}$ using \cite[Algorithm~2]{solomon2020soft}
			\IF{$\underline{\alpha}$ satisfies
				\eqref{eq:first-order-validity}}
			\STATE $q\leftarrow q+1$
			\STATE
			$\underline{\widetilde{c}}
			\leftarrow
			\underline{c}^{*}
			\oplus
			\underline{e}(\underline{\alpha})$
			\IF{$\mathbf{H}
				\underline{\widetilde{c}}^{\mathsf{T}}
				=\underline{0}$}
			\RETURN
			$\underline{\widehat{c}}
			\leftarrow\underline{\widetilde{c}}$
			\ENDIF
			\ENDIF
			\ENDWHILE
			\RETURN \textsc{Abandon}
		\end{algorithmic}
	\end{algorithm}

	\subsection{ORBGRAND-ISI and CDF-ORBGRAND-ISI}
	\label{subsec:orb-isi}

	Although SGRAND-ISI achieves ML decoding under complete querying, it requires exact sequence reliabilities for online generation. ORBGRAND-ISI instead follows ORBGRAND~\cite{duffy2022ordered} and
	replaces the exact sequence reliability values with the corresponding
	ranks, i.e.,
	$\gamma(\underline{y},\mathcal{B}_v)=r(\mathcal{B}_v)$. To further improve decoding performance, CDF-ORBGRAND-ISI maps the
	normalized reliability ranks through the inverse CDF and employs
	$\gamma(\underline{y},\mathcal{B}_v)
	=\Psi^{-1}\left(r(\mathcal{B}_v)/(V_1+1)\right)$,
	where $\Psi(\cdot)$ denotes the CDF of the sequence reliability of an
	error burst.
	
	\begin{figure}[htbp]
		\centering 
		\includegraphics[width=0.46\textwidth]{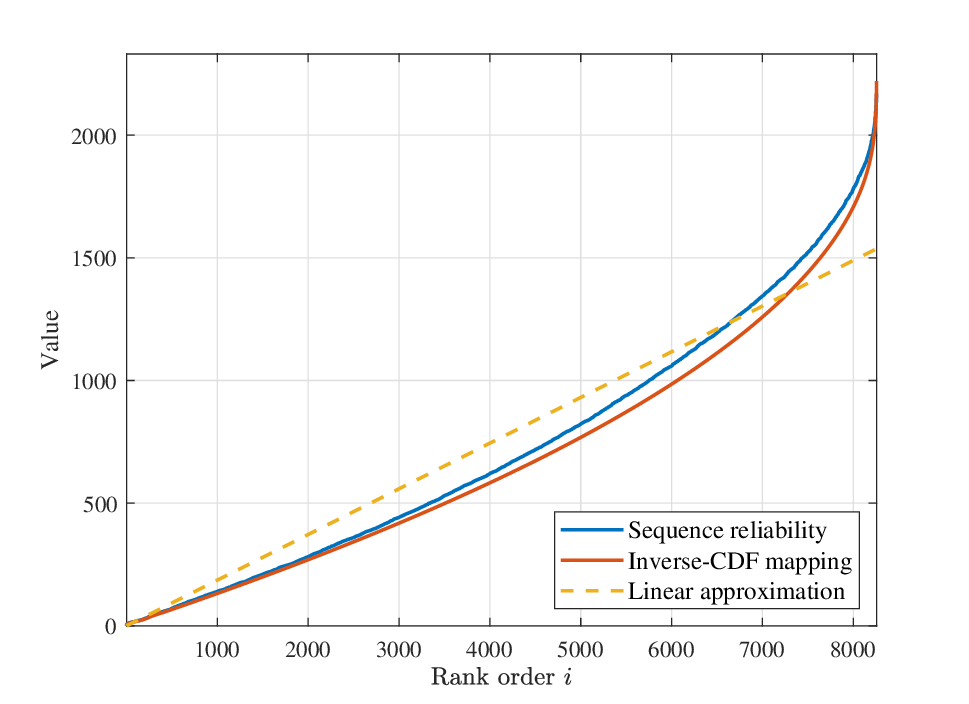}
		\caption{Comparison of the sequence reliability, inverse-CDF
			mapping, and linear approximation for a BPSK sequence of length $N=128$ over the first-order ISI channel with $(h_0,h_1)=(\sqrt{0.9},\sqrt{0.1})$ at an SNR of $6$ dB. }
		\label{fig:Psi}
	\end{figure}
	
	Fig.~\ref{fig:Psi} illustrates the relationship between sequence reliability and its rank-based approximations for a representative BPSK sequence of length $N=128$.
	After sorting the sequence reliabilities of all
	$V_1=N(N+1)/2=8256$ possible error bursts in ascending order, we plot
	the $r$-th smallest sequence reliability, the inverse-CDF value
	$\Psi^{-1}\left(r/(V_1+1)\right)$, and the corresponding linear
	approximation; see, e.g., \cite{duffy2022ordered} for details on the
	linear approximation underlying ORBGRAND. These three curves correspond
	to the metrics employed by SGRAND-ISI, CDF-ORBGRAND-ISI, and
	ORBGRAND-ISI, respectively. It can be seen that the inverse-CDF mapping more closely follows the exact sequence reliability values than the linear approximation. The numerical results in Section~\ref{sec:numerical-results} further show that this mapping leads to a clear performance gain over ORBGRAND-ISI.
	
	\section{Extension to Higher-Order ISI Channels}
	\label{sec:higher-order-ISI}

	We now extend the GRAND-ISI framework to an ISI channel of order
	$L\geq2$. Since flipping the symbol at position $s$ can affect the
	sequence metric terms indexed from $s$ to $s+L$, error positions
	separated by at most $L$ may have coupled contributions. This motivates
	the following generalization of the error burst decomposition.

	\begin{definition}[Order-$L$ error burst decomposition]
		\label{defn:higher-order-error-burst}
		For a nonzero EP $\underline{e}\in\{0,1\}^{N}$, its support admits
		a unique partition into pairwise disjoint nonempty subsets
		\begin{align}
			\operatorname{supp}(\underline{e})
			=
			\bigcup_{k=1}^{\zeta_L(\underline{e})}
			\mathcal{B}^{(L)}_k(\underline{e}).
		\end{align}
		Writing
		$\mathcal{B}^{(L)}_k(\underline{e})
		=\{b_{k,1}<\cdots<b_{k,m_k}\}$,
		each subset satisfies
		\begin{align}
			b_{k,j+1}-b_{k,j}
			\leq L,
			\qquad
			j=1,\ldots,m_k-1.
		\end{align}
		The subsets are indexed such that
		\begin{align}
			b_{1,1}<b_{2,1}<\cdots<
			b_{\zeta_L(\underline{e}),1},
		\end{align}
		and consecutive subsets satisfy
		\begin{align}
			b_{k+1,1}-b_{k,m_k}>L,
			\qquad
			k=1,\ldots,\zeta_L(\underline{e})-1.
		\end{align}
		Each $\mathcal{B}^{(L)}_k(\underline{e})$ is called an order-$L$
		error burst, and $\zeta_L(\underline{e})$ denotes the number of
		order-$L$ error bursts. For $\underline{e}=\underline{0}$, the
		decomposition is empty and $\zeta_L(\underline{0})=0$.
	\end{definition}

	For $L=1$, Definition~\ref{defn:higher-order-error-burst} reduces to
	the first-order error burst decomposition in Definition~\ref{defn:error-burst}. For $L\geq2$, however, an
	order-$L$ error burst need not be a contiguous interval. The following
	example illustrates the order-$L$ error burst decomposition.
	
	\begin{example}
		Consider an ISI channel of order $L=2$ and the EP
		$\underline{e}=[1,0,1,0,0,1,1,0]$, whose support is
		$\operatorname{supp}(\underline{e})=\{1,3,6,7\}$. Its
		order-$2$ error burst decomposition is
		\begin{align}
			\mathcal{B}^{(2)}_1(\underline{e})
			&=\{1,3\},&
			\mathcal{B}^{(2)}_2(\underline{e})
			&=\{6,7\},
		\end{align}
		and hence $\zeta_2(\underline{e})=2$.
	\end{example}
	
	\subsection{Properties of Sequence Reliability}
	\label{subsec:properties-higher-order-ISI}
	
	For higher-order ISI channels, sequence reliability retains the
	definition in Definition~\ref{defn:sequence-reliability}, with
	$\Gamma(\underline{x},\underline{y})$ evaluated according to the
	order-$L$ sequence metric defined in \eqref{eq:sequence-metric}.
	Although its evaluation is more involved than in the first-order case,
	sequence reliability remains additive across sufficiently separated
	index sets.
	
	\begin{lemma}
		\label{lemma:higher-order-additivity}
		A nonempty index set $\mathcal{S}\subseteq\{1,\ldots,N\}$ is called decomposable under the order-$L$ ISI model if it can be written as
		$\mathcal{S}=\mathcal{S}_1\cup\mathcal{S}_2$, where $\mathcal{S}_1$ and $\mathcal{S}_2$ are nonempty and satisfy
		\begin{align}
			\mathcal{S}_1\cap\mathcal{S}_2
			=\emptyset, \quad\min(\mathcal{S}_2)-\max(\mathcal{S}_1)
			>L.
		\end{align}
		Then, we obtain
		\begin{align}\label{eq:higher-order-additivity}
			\operatorname{Rel}(\mathcal{S})
			=
			\operatorname{Rel}(\mathcal{S}_1)
			+
			\operatorname{Rel}(\mathcal{S}_2).
		\end{align}
	\end{lemma}
	
	\begin{IEEEproof}
		See Appendix~\ref{proof:higher-order-additivity}.
	\end{IEEEproof}
	
	Repeated application of Lemma~\ref{lemma:higher-order-additivity} to the order-$L$ error burst decomposition gives
	\begin{align}\label{eq:higher-order-error-burst-additivity}
		\operatorname{Rel}\bigl(
		\operatorname{supp}(\underline{e})
		\bigr)
		=
		\sum_{k=1}^{\zeta_L(\underline{e})}
		\operatorname{Rel}\left(
		\mathcal{B}^{(L)}_k(\underline{e})
		\right),
	\end{align}
	for every $\underline{e}\in\{0,1\}^{N}$, with the empty-sum
	convention for $\underline{e}=\underline{0}$.
	
	To distinguish the possible structures within an order-$L$ error
	burst, consider
	$\mathcal{B}^{(L)}=\{b_1<\cdots<b_m\}$. It is called non-decomposable if
	\begin{align}
		b_{j+1}-b_j=1,
		\qquad j=1,\ldots,m-1,
	\end{align}
	and partially decomposable otherwise. In the latter case, $2\leq b_{j+1}-b_j\leq L$ for at least one $j$. Hence, a partially decomposable error burst is noncontiguous, but its constituent error positions remain coupled through the ISI channel.
	
	Unlike the decomposable case in
	Lemma~\ref{lemma:higher-order-additivity}, splitting a partially
	decomposable error burst across an internal gap does not generally
	yield additive sequence reliability, since the corresponding summands
	in \eqref{eq:sequence-metric} may overlap. Consequently, the
	first-order recursion in Lemma~\ref{lemma:2} does not directly apply,
	and the sequence reliability is evaluated from
	Definition~\ref{defn:sequence-reliability}.

	\subsection{GRAND-ISI for Higher-Order ISI Channels}
	\label{subsec:grand-higher-order-ISI}
	
	Based on the order-$L$ error burst decomposition, define 
	\begin{align}\label{eq:higher-order-aggregate-reliability}
		\rho_{\gamma}^{(L)}(\underline{y},\underline{e})
		\triangleq
		\sum_{k=1}^{\zeta_L(\underline{e})}
		\gamma\!\left(
		\underline{y},
		\mathcal{B}^{(L)}_k(\underline{e})
		\right),
	\end{align}
	with the empty-sum convention for $\underline{e}=\underline{0}$. For $L=1$, \eqref{eq:higher-order-aggregate-reliability} reduces to the
	first-order aggregate metric in \eqref{eq:aggregate-reliability}.
	
	The GRAND-ISI decoding rule in \eqref{eq:general-GRAND-ISI-decoder} remains unchanged after replacing $\rho_{\gamma}$ with $\rho_{\gamma}^{(L)}$. The first-order variants therefore extend formally by applying the same choices of $\gamma$ to order-$L$ error bursts: SGRAND-ISI uses their exact sequence reliabilities, ORBGRAND-ISI uses their ranks, and CDF-ORBGRAND-ISI uses the corresponding inverse-CDF mapping values.
	
	For SGRAND-ISI, $\gamma(\underline{y},\mathcal{S})=\operatorname{Rel}(\mathcal{S})$. Hence, \eqref{eq:higher-order-error-burst-additivity} gives
	\begin{align}\label{eq:higher-order-sgrand-metric}
		\rho_{\gamma}^{(L)}(\underline{y},\underline{e})
		=
		\operatorname{Rel}\!\left(
		\operatorname{supp}(\underline{e})
		\right).
	\end{align}
	As in the first-order case, SGRAND-ISI satisfies the following
	ML-equivalence result under the order-$L$ ISI model.
	
	\begin{corollary}
		\label{cor:higher-order-sgrand-ml}
		Under the order-$L$ ISI model, SGRAND-ISI is equivalent to ML
		decoding under complete querying, i.e., when $Q=2^N$.
	\end{corollary}
	
	\begin{IEEEproof}
		By \eqref{eq:higher-order-sgrand-metric}, the aggregate metric of
		every EP equals its exact sequence reliability. The result therefore
		follows from the same argument as
		Theorem~\ref{thm:SGRAND-ISI-ML}.
	\end{IEEEproof}
	
	\subsection{Approximation Strategy}
	\label{subsec:higher-order-approximate}

	Although the GRAND-ISI framework extends formally to higher-order
	ISI channels, its direct implementation requires evaluating the
	sequence reliabilities of all order-$L$ error bursts. For $L\geq2$,
	the large number of partially decomposable error bursts makes this
	computation prohibitive. We therefore retain all non-decomposable
	error bursts, while restricting partially decomposable error bursts
	to those containing at most $g$ positions.
	
	Specifically, for $g\in\{1,\ldots,N\}$, define the retained error burst set as
	\begin{align}
		\label{eq:retained-order-L-bursts}
		\mathscr{B}_{L,g}
		\triangleq
		\left\{
		\mathcal{B}\,\middle|\,
		\begin{aligned}
			&\mathcal{B}\text{ is an order-}L\text{ error burst,}\\
			&\text{and }\mathcal{B}\text{ is non-decomposable or }
			|\mathcal{B}|\leq g
		\end{aligned}
		\right\}.
	\end{align}
	Thus, $g=1$ retains only the non-decomposable error bursts, whereas
	$g=N$ recovers the complete set of order-$L$ error bursts. The SGRAND-ISI restricted to $\mathscr{B}_{L,g}$ is referred
	to as restricted SGRAND-ISI-$g$ (R-SGRAND-ISI-$g$). Its approximation
	results from omitting partially decomposable error bursts containing
	more than $g$ positions.
	
	We next describe a realization of R-SGRAND-ISI-$g$. Let
	$V_{L,g}\triangleq|\mathscr{B}_{L,g}|$.
	Index the retained error bursts as
	$\mathcal{B}_1,\ldots,\mathcal{B}_{V_{L,g}}$
	in increasing order of their sequence reliabilities, such that
	$\operatorname{Rel}(\mathcal{B}_1)\leq\cdots\leq\operatorname{Rel}(\mathcal{B}_{V_{L,g}})$,
	with ties broken arbitrarily. A burst-selection vector
	$\underline{\alpha}
	=[\alpha_1,\ldots,\alpha_{V_{L,g}}]
	\in\{0,1\}^{V_{L,g}}$
	selects $\mathcal{B}_v$ whenever $\alpha_v=1$. A burst-selection vector $\underline{\alpha}$ is declared valid if, for
	any distinct $u,v\in\operatorname{supp}(\underline{\alpha})$,
	\begin{align}
		\mathcal{B}_u\cap\mathcal{B}_v
		&=
		\emptyset,
		\qquad
		\mathcal{B}_u\cup\mathcal{B}_v
		\notin
		\mathscr{B}_{L,g}.
		\label{eq:restricted-validity}
	\end{align}
	The first condition prevents overlapping retained error bursts, whereas
	the second avoids redundant representations of an error burst already
	contained in $\mathscr{B}_{L,g}$. Each valid burst-selection vector induces an EP
	$\underline{e}(\underline{\alpha})$ with
	\begin{align}\label{eq:restricted-selection-to-EP}
		\operatorname{supp}
		\bigl(\underline{e}(\underline{\alpha})\bigr)
		=
		\bigcup_{v\in\operatorname{supp}(\underline{\alpha})}
		\mathcal{B}_v.
	\end{align}
	Distinct valid burst-selection vectors may induce the same EP.
	Moreover, because partially decomposable error bursts of size greater than $g$ are omitted, their supports may be represented by several retained error bursts. These error bursts may remain coupled through ISI, so the sum of their sequence reliabilities generally only approximates the exact sequence reliability of the omitted error burst.
	
	Using the retained error bursts as events and their sequence
	reliabilities as weights, burst-selection vectors are generated online
	using \cite[Algorithm~2]{solomon2020soft}. Invalid vectors are
	discarded, while each valid vector is mapped to $\underline{e}(\underline{\alpha})$ through \eqref{eq:restricted-selection-to-EP}, and
	$\underline{c}^{*}\oplus\underline{e}(\underline{\alpha})$ is subjected
	to a syndrome test. Decoding stops when a codeword is found or the abandonment threshold $Q$ is reached.
	
	Correspondingly, replacing the exact sequence reliabilities of the
	retained error bursts with their ranks or the corresponding
	inverse-CDF mapping values yields R-ORBGRAND-ISI-$g$ and
	R-CDF-ORBGRAND-ISI-$g$, respectively, following the same construction
	as in Section~\ref{subsec:orb-isi}.

	\section{Extension to Higher-Order Modulation}
	\label{sec:higher-order-modulation}
	
	We next extend the GRAND-ISI framework from BPSK to higher-order
	modulation. Symbol-level GRAND has previously been considered for
	higher-order modulation \cite{an2022keep}, and ORBGRAND was subsequently
	adapted to exploit symbol-level soft information directly
	\cite{an2023soft}. Building on this symbol-level perspective, we extend
	the GRAND-ISI framework to higher-order modulation in the presence of
	ISI, where symbol substitutions are coupled across channel uses through
	the ISI sequence metric. To focus on the impact of the modulation order,
	we consider a first-order ISI channel in this section. Unlike BPSK, an
	erroneous modulation symbol may take multiple alternative values.
	Therefore, a symbol error event must specify both the affected symbol
	positions and their replacement symbols.
	
	\subsection{System Model}
	\label{subsec:higher-modulation-system-model}
	
	Let $\mathcal{X}$ denote a constellation of size
	$J\triangleq|\mathcal{X}|=2^m$, where each modulation symbol carries
	$m$ coded bits. The constellation is normalized such that
	$\mathbb{E}[|\mathsf{S}_i|^2]=1$. A modulation mapping
	$\mu:\{0,1\}^{m}\rightarrow\mathcal{X}$ maps every group of $m$ coded
	bits to one constellation symbol.  Assuming that $m$ divides
	$N$, the binary codeword $\underline{c}(w)\in\{0,1\}^{N}$ is mapped
	to a symbol sequence $\underline{s}(w)=[s_1(w),\ldots,s_{N_s}(w)]
	\in\mathcal{X}^{N_s}$, where $N_s\triangleq N/m$.
	
	The transmitted symbol sequence passes through a first-order ISI
	channel described by
	\begin{align}\label{eq:mary-first-order-channel}
		\mathsf{Y}_i
		=
		h_0\mathsf{S}_i
		+
		h_1\mathsf{S}_{i-1}
		+
		\mathsf{Z}_i,
		\qquad i=1,\ldots,N_s,
	\end{align}
	where $\{\mathsf{Z}_i\}$ are i.i.d. circularly
	symmetric complex Gaussian random variables with
	$\mathsf{Z}_i\sim\mathcal{CN}(0,2\sigma^2)$, and
	$h_0,h_1\in\mathbb{C}$ are normalized channel taps satisfying
	$|h_0|^2+|h_1|^2=1$. Consequently, the SNR is $1/(2\sigma^2)$.
	
	For a symbol sequence
	$\underline{s}=[s_1,\ldots,s_{N_s}]\in\mathcal{X}^{N_s}$, define the
	sequence metric
	\begin{align}\label{eq:mary-sequence-metric}
		\Gamma_{\mathrm{mod}}(\underline{s},\underline{y})
		\triangleq
		\sum_{i=1}^{N_s}
		\ln
		p_{\mathsf{Y}_i|\underline{\mathsf{S}}_{i-1:i}}
		\bigl(
		y_i|\underline{s}_{i-1:i}
		\bigr).
	\end{align}
	The Viterbi detector produces the hard-decision symbol sequence
	$\underline{s}^{*}=\arg\max_{\underline{s}\in\mathcal{X}^{N_s}}\Gamma_{\mathrm{mod}}(\underline{s},\underline{y})$.
	
	\subsection{Symbol Error Burst and Sequence Reliability}
	\label{subsec:mary-sequence-reliability}
	
	Unlike BPSK, the erroneous symbol positions alone do not uniquely
	specify a symbol EP, since each erroneous symbol may be
	replaced by any of the other $J-1$ constellation symbols. Let
	\begin{align}
		\mathcal{D}
		=
		\{i_1,\ldots,i_d\}
		\subseteq
		\{1,\ldots,N_s\},
		\qquad
		i_1<\cdots<i_d,
	\end{align}
	denote a given set of erroneous symbol positions. For a particular
	choice of replacement symbols, let
	\begin{align}
		\underline{u}
		=
		[u_1,\ldots,u_d],
		\qquad
		u_k\in
		\mathcal{X}\setminus\{s_{i_k}^{*}\},
		\quad
		k=1,\ldots,d,
	\end{align}
	denote the corresponding replacement vector. The symbol EP specified by $\mathcal{D}$ and $\underline{u}$ is defined as
	\begin{align}\label{eq:symbol-error-pattern}
		\mathcal{P}(\mathcal{D},\underline{u})
		\triangleq
		\left\{
		(i_k,u_k):
		k=1,\ldots,d
		\right\}.
	\end{align}
	
	For a fixed position set $\mathcal{D}$, all 
	symbol EPs form the set
	\begin{align}\label{eq:symbol-error-pattern-set}
		\Omega(\mathcal{D})
		\triangleq
		\left\{
		\mathcal{P}(\mathcal{D},\underline{u}):
		u_k\in
		\mathcal{X}\setminus\{s_{i_k}^{*}\},
		\quad
		k=1,\ldots,d
		\right\}.
	\end{align}
	Thus,
	$\mathcal{P}(\mathcal{D},\underline{u})\in\Omega(\mathcal{D})$
	represents one particular assignment of replacement symbols to the
	positions in $\mathcal{D}$. Since each of the $d$ positions admits $J-1$ possible replacement symbols, we have $|\Omega(\mathcal{D})|=(J-1)^d$.
	
	For a nonempty position set $\mathcal{D}$, there exists a unique
	partition into maximal contiguous intervals, $\mathcal{D}=\bigcup_{k=1}^{\zeta(\mathcal{D})}\mathcal{B}_k$,
	where each $\mathcal{B}_k$ is a contiguous interval. The intervals are
	indexed such that
	\begin{align}
		\min(\mathcal{B}_1)
		<
		\min(\mathcal{B}_2)
		<
		\cdots
		<
		\min(\mathcal{B}_{\zeta(\mathcal{D})}),
	\end{align}
	and consecutive intervals satisfy
	\begin{align}
		\max(\mathcal{B}_k)+1
		<
		\min(\mathcal{B}_{k+1}),
		\qquad
		k=1,\ldots,\zeta(\mathcal{D})-1.
	\end{align}
	
	For a symbol EP
	$\mathcal{P}(\mathcal{D},\underline{u})$, the symbol error burst
	associated with $\mathcal{B}_k$ is defined as
	\begin{align}
		\mathcal{E}_k
		\triangleq
		\left\{
		(i_j,u_j):
		i_j\in\mathcal{B}_k
		\right\},
		\qquad
		k=1,\ldots,\zeta(\mathcal{D}).
		\label{eq:symbol-error-burst-event}
	\end{align}
	Accordingly, every symbol EP admits the unique
	decomposition $\mathcal{P}(\mathcal{D},\underline{u})
	=\bigcup_{k=1}^{\zeta(\mathcal{D})}\mathcal{E}_k$.
	For any symbol EP
	$\mathcal{P}(\mathcal{D},\underline{u})\in\Omega(\mathcal{D})$,
	let
	$\mathcal{R}\bigl(
	\underline{s}^{*};
	\mathcal{P}(\mathcal{D},\underline{u})
	\bigr)$
	denote the symbol sequence obtained by applying the replacements
	specified by $\mathcal{P}(\mathcal{D},\underline{u})$ to
	$\underline{s}^{*}$. Its $i$-th symbol is given by
	\begin{align}\label{eq:symbol-pattern-replacement}
		\left[
		\mathcal{R}
		\left(
		\underline{s}^{*};
		\mathcal{P}(\mathcal{D},\underline{u})
		\right)
		\right]_i
		\triangleq
		\begin{cases}
			u_k, & i=i_k,\quad k=1,\ldots,d,\\
			s_i^{*}, & i\notin\mathcal{D}.
		\end{cases}
	\end{align}
	The sequence reliability of
	$\mathcal{P}(\mathcal{D},\underline{u})$ is defined as
	\begin{align}\label{eq:mary-sequence-reliability}
		\operatorname{Rel}
		\left(
		\mathcal{P}(\mathcal{D},\underline{u})
		\right)
		\triangleq
		\Gamma_{\mathrm{mod}}
		\left(
		\underline{s}^{*},
		\underline{y}
		\right)
		-
		\Gamma_{\mathrm{mod}}
		\left(
		\mathcal{R}
		\left(
		\underline{s}^{*};
		\mathcal{P}(\mathcal{D},\underline{u})
		\right),
		\underline{y}
		\right).
	\end{align}
	Since $\underline{s}^{*}$ maximizes the sequence metric,
	$\operatorname{Rel}
	\bigl(\mathcal{P}(\mathcal{D},\underline{u})\bigr)\geq0$,
	and a smaller sequence reliability corresponds to a more likely
	symbol EP.
	
	As in the BPSK case, the sequence reliability is additive for
	sufficiently separated symbol EPs under the first-order
	ISI model.
	
	\begin{lemma}\label{lemma:mary-additivity}
		Let
		$\mathcal{Q}_1\in\Omega(\mathcal{D}_1)$ and
		$\mathcal{Q}_2\in\Omega(\mathcal{D}_2)$
		be two symbol EPs. If $\max(\mathcal{D}_1)+1<\min(\mathcal{D}_2)$, then we obtain
		\begin{align}\label{eq:mary-additivity}
			\operatorname{Rel}
			\left(
			\mathcal{Q}_1\cup\mathcal{Q}_2
			\right)
			=
			\operatorname{Rel}(\mathcal{Q}_1)
			+
			\operatorname{Rel}(\mathcal{Q}_2).
		\end{align}
	\end{lemma}
	\begin{IEEEproof}
		See Appendix~\ref{app:proof-mary-additivity}.
	\end{IEEEproof}
	
	Since the symbol error bursts are separated by at least
	one unchanged symbol, Lemma~\ref{lemma:mary-additivity} gives
	$\operatorname{Rel}\left(\mathcal{P}(\mathcal{D},\underline{u})\right)=\sum_{k=1}^{\zeta(\mathcal{D})}\operatorname{Rel}(\mathcal{E}_k)$.
	Thus, the sequence reliability of an arbitrary symbol EP can
	be obtained from the sequence reliabilities of its constituent symbol
	error bursts.
	
	\subsection{GRAND-ISI for Higher-Order Modulation}
	\label{subsec:mary-decoding}
	
	We next extend SGRAND-ISI to higher-order modulation based on the
	symbol error bursts defined above. Consider a contiguous interval
	$\mathcal{B}$ of length $\ell$. Since each position in $\mathcal{B}$
	admits $J-1$ replacement symbols, there are $(J-1)^\ell$ possible
	symbol error bursts associated with $\mathcal{B}$. As there are
	$N_s-\ell+1$ contiguous intervals of length $\ell$, the total number
	of possible symbol error bursts is
	\begin{align}
		V_{\mathrm{mod}}
		\triangleq
		\sum_{\ell=1}^{N_s}
		(N_s-\ell+1)(J-1)^\ell.
		\label{eq:number-symbol-error-burst-events}
	\end{align}
	
	Collect all possible symbol error bursts in $\mathscr{E}\triangleq\left\{\mathcal{E}_1,\ldots,\mathcal{E}_{V_{\mathrm{mod}}}\right\}$, where they are indexed such that $\operatorname{Rel}(\mathcal{E}_1)\leq\cdots\leq\operatorname{Rel}(\mathcal{E}_{V_{\mathrm{mod}}})$, with ties broken arbitrarily. For each
	$\mathcal{E}_v\in\mathscr{E}$, let $\mathcal{B}_v$ denote its
	associated contiguous interval.
	
	A symbol EP is specified by a burst-selection
	vector $\underline{\alpha}=[\alpha_1,\ldots,\alpha_{V_{\mathrm{mod}}}]\in\{0,1\}^{V_{\mathrm{mod}}}$,
	where $\alpha_v=1$ indicates that $\mathcal{E}_v$ is selected.
	According to the symbol error burst decomposition defined above, a
	burst-selection vector is valid if any two selected bursts are
	separated by at least one unchanged symbol. Specifically, for any
	distinct $u$ and $v$ with $\alpha_u=\alpha_v=1$ and
	$\min(\mathcal{B}_u)<\min(\mathcal{B}_v)$,
	\begin{align}
		\max(\mathcal{B}_u)+1
		<
		\min(\mathcal{B}_v).
		\label{eq:mary-valid-selection}
	\end{align}
	The all-zero vector is also regarded as valid.
	
	For a valid $\underline{\alpha}$, the corresponding symbol EP is
	\begin{align}
		\mathcal{P}_{\underline{\alpha}}
		\triangleq
		\bigcup_{v:\alpha_v=1}
		\mathcal{E}_v.
	\end{align}
	By Lemma~\ref{lemma:mary-additivity}, its sequence reliability is
	\begin{align}
		\operatorname{Rel}
		\left(
		\mathcal{P}_{\underline{\alpha}}
		\right)
		=
		\sum_{v:\alpha_v=1}
		\operatorname{Rel}(\mathcal{E}_v).
		\label{eq:mary-selection-reliability}
	\end{align}
	
	SGRAND-ISI generates valid burst-selection vectors in increasing order
	of \eqref{eq:mary-selection-reliability}, following the online
	generation principle of SGRAND
	\cite[Algorithm~2]{solomon2020soft}. The all-zero vector is queried
	first and corresponds to $\underline{s}^{*}$. For each subsequent valid
	$\underline{\alpha}$, the symbol sequence
	$\mathcal{R}(\underline{s}^{*};\mathcal{P}_{\underline{\alpha}})$ is
	inversely mapped to a binary sequence and subjected to a syndrome test.
	Decoding terminates when a codeword is found or the abandonment
	threshold is reached.
	
	Under complete querying, SGRAND-ISI achieves ML decoding. Every
	symbol sequence differs from the hard-decision sequence
	$\underline{s}^{*}$ at a unique set of symbol positions and therefore
	corresponds to a unique symbol EP. By the unique symbol
	error burst decomposition, it is further represented by a unique valid
	burst-selection vector. Hence, complete querying covers all symbol sequences. Moreover, by \eqref{eq:mary-sequence-reliability},
	the sequence reliability with a symbol sequence equals the constant
	$\Gamma_{\mathrm{mod}}(\underline{s}^{*},\underline{y})$ minus its
	sequence metric. Therefore, ordering symbol sequences in increasing sequence
	reliability is equivalent to ordering them in decreasing sequence
	metric, and the first codeword identified by SGRAND-ISI is an ML
	codeword.
	
	ORBGRAND-ISI and CDF-ORBGRAND-ISI follow the same burst-selection
	framework. ORBGRAND-ISI replaces the exact sequence reliability values
	with the corresponding reliability ranks, while CDF-ORBGRAND-ISI
	further applies inverse-CDF companding to the normalized ranks.
	
	\subsection{Pruning Strategy}
	\label{subsec:mary-pruning}
	
	Although the unpruned GRAND-ISI framework preserves all possible
	symbol error bursts, the number of such bursts grows rapidly with the
	modulation order and burst length. To reduce the resulting storage and
	online generation complexity, we introduce two pruning parameters:
	$T$, the number of retained replacement symbols at each symbol
	position, and $\ell_{\max}$, the maximum length of the symbol error
	bursts whose sequence reliabilities are explicitly evaluated.
	
	For each symbol position $i$, order the $J-1$ alternative constellation
	symbols as $\operatorname{Rel}\!\left(\{(i,a_{i,1})\}\right)\leq\cdots\leq\operatorname{Rel}\!\left(\{(i,a_{i,J-1})\}\right)$.
	For a prescribed $T\leq J-1$, only the $T$ most likely replacement
	symbols are retained, forming $\mathcal{X}_i^{(T)}\triangleq\{a_{i,1},\ldots,a_{i,T}\}$.
	This pruning is motivated by the nearest-neighbor selection used in
	symbol-level ORBGRAND for higher-order modulation \cite{an2023soft},
	where only a small number of likely neighboring constellation symbols
	are considered for substitution. Here, the same principle is adapted
	to the ISI setting using the corresponding symbol likelihoods.
	For a contiguous interval $\mathcal{B}$ of length $\ell$, restricting
	the replacement at each position $i\in\mathcal{B}$ to
	$\mathcal{X}_i^{(T)}$ leaves $T^\ell$ possible symbol error bursts.
	
	We further explicitly evaluate only symbol error bursts of length at
	most $\ell_{\max}$. The total number of retained elementary symbol
	error bursts is therefore
	\begin{align}
		V_{T,\ell_{\max}}
		\triangleq
		\sum_{\ell=1}^{\ell_{\max}}
		(N_s-\ell+1)T^\ell.
		\label{eq:number-retained-symbol-error-burst-events}
	\end{align}
	Collect these symbol error bursts in
	$\mathscr{E}_{T,\ell_{\max}}
	\triangleq
	\{\mathcal{E}_1,\ldots,\mathcal{E}_{V_{T,\ell_{\max}}}\}$,
	indexed in nondecreasing order of their sequence reliabilities, with ties
	broken arbitrarily.
	
	The resulting pruned decoder is referred to as P-SGRAND-ISI. The retained symbol error bursts form the elementary units for decoding, and a burst-selection vector
	$\underline{\alpha}\in\{0,1\}^{V_{T,\ell_{\max}}}$
	specifies a combination of retained symbol error bursts. A selection vector is valid if the selected symbol error bursts are pairwise disjoint and no pair can be merged into another retained symbol error burst in $\mathscr{E}_{T,\ell_{\max}}$. Symbol error bursts longer than $\ell_{\max}$ can still be formed by combining retained shorter symbol error bursts, with their sequence reliabilities approximated by sums of the corresponding retained symbol error burst sequence reliabilities. P-SGRAND-ISI then follows the same online generation and codeword-testing procedure as SGRAND-ISI.
	
	
	Because pruning may exclude some replacement symbols and the
	reliabilities of longer symbol error bursts are approximated, the exact likelihood
	ordering is not generally preserved, and P-SGRAND-ISI is therefore not
	guaranteed to achieve ML decoding. The same pruning strategy is applied to obtain P-ORBGRAND-ISI and P-CDF-ORBGRAND-ISI.
	
	\section{Numerical Results}\label{sec:numerical-results}
	
	In this section, we present numerical simulation results for the
	proposed GRAND-ISI algorithms. We first consider BPSK modulation and
	present results over both first-order and higher-order ISI channels.
	We then extend the evaluation to higher-order modulation to demonstrate
	the applicability of the proposed framework beyond binary signaling
	and to assess the effectiveness of the proposed pruning strategy.
	Finally, we analyze the computational complexity of the proposed
	decoding algorithms.
	
	\subsection{BPSK Signaling}\label{subsec:simulation-bpsk}
	
	\subsubsection{First-Order ISI}
	
	We first consider BPSK modulation over first-order ISI channels
	using a BCH$(127,113)$ code. For comparison, we consider the
	following decoding algorithms:
	\begin{itemize}
		\item SGRAND-ISI;
		\item ORBGRAND-ISI;
		\item CDF-ORBGRAND-ISI;
		\item ORBGRAND-AI with $B=4$ and $B=8$ in
		\cite{duffy2023using,duffy2025decoding}; Here, ORBGRAND-AI is adapted directly to the original ISI channel by
		computing its block reliabilities from the corresponding ISI likelihood. \\
		For all the above four decoding algorithms we set the maximum number of queries $Q=1\times 10^4$.
		\item ORBGRAND without considering channel memory, with $Q=1.5\times 10^5$;
		\item Turbo equalization.\footnote{The turbo-equalization baseline follows the conventional
			framework \cite{douillard1995iterative,tuchler2002turbo}, with iterative
			soft-information exchange between a BCJR equalizer
			\cite{bahl1974optimal} and a Chase-Pyndiah-type SISO BCH decoder
			\cite{pyndiah1998near}. The decoder uses the five least reliable
			code-bit positions, yielding $2^5=32$ test patterns, with
			bounded-distance decoding radius $2$, extrinsic scaling factor $0.7$,
			and five turbo iterations. A length-$127$ random block interleaver is
			used, with extrinsic LLRs deinterleaved before BCH decoding and
			reinterleaved before feedback to the BCJR equalizer.}
	\end{itemize}
	We also include a baseline of genie-aided ML lower bound.


	For $h_0=\sqrt{0.9}$ and $h_1=\sqrt{0.1}$, the decoding performance of the BCH$(127,113)$ code is shown in Fig.~\ref{fig:bpsk_oneISI_09}. Conventional ORBGRAND, which ignores the channel memory, exhibits a substantial performance loss. In contrast, exploiting the channel correlation through ORBGRAND-AI significantly improves the decoding performance, with $B=8$ providing a noticeable gain over $B=4$. The proposed GRAND-ISI algorithms further improve the decoding performance. At a BLER of $10^{-3}$, ORBGRAND-ISI, CDF-ORBGRAND-ISI, and SGRAND-ISI provide approximately $0.3$ dB, $0.5$ dB, and $0.6$ dB gains over ORBGRAND-AI with $B=8$, respectively. Turbo equalization
	achieves performance comparable to ORBGRAND-ISI, while
	CDF-ORBGRAND-ISI and SGRAND-ISI provide additional gains. In particular, SGRAND-ISI approaches the genie-aided ML lower bound in the low-BLER regime, demonstrating the effectiveness of exploiting the sequence reliability induced by the ISI channel.
	
	\begin{figure}[htbp]
		\centering 
		\includegraphics[width=0.46\textwidth]{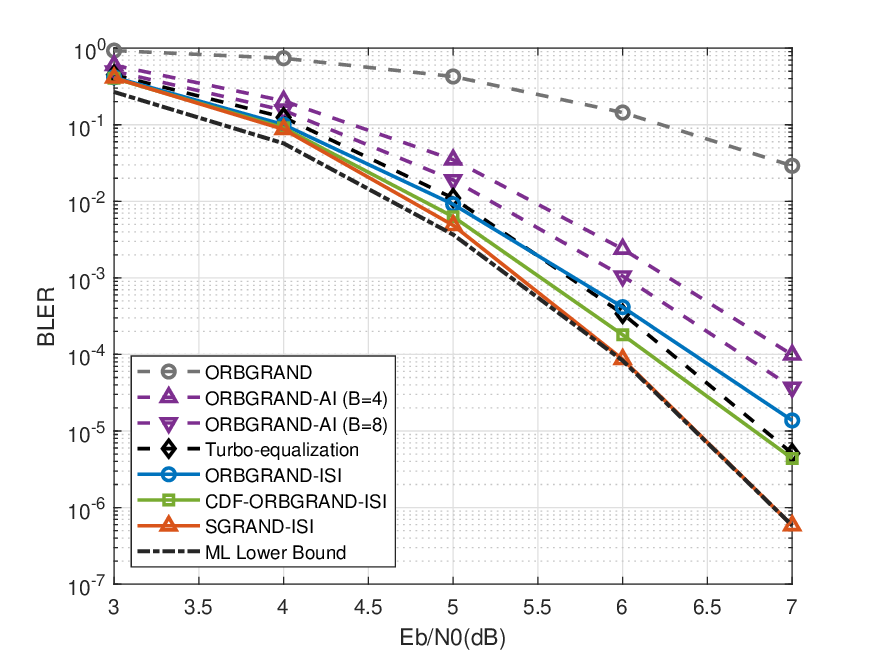}
		\caption{BLER performance with BPSK modulation over a first-order ISI
			channel with $h_0=\sqrt{0.9}$ and $h_1=\sqrt{0.1}$.}
		\label{fig:bpsk_oneISI_09}
	\end{figure}
	
	The corresponding average numbers of queries are summarized in
	Table~\ref{tab:bpsk_oneISI_09_query}. SGRAND-ISI requires the fewest
	queries, followed by CDF-ORBGRAND-ISI, ORBGRAND-ISI, and ORBGRAND-AI.
	Nevertheless, SGRAND-ISI has the highest online processing complexity
	because its EPs are dynamically generated according to the
	exact sequence reliabilities. A more detailed complexity comparison
	between ORBGRAND-ISI and ORBGRAND-AI is provided in
	Section~\ref{subsec:complexity-analysis}.
	
	\begin{table}[htbp]
		\centering
		\caption{Average number of queries for the setting in Fig.~\ref{fig:bpsk_oneISI_09}.}
		\label{tab:bpsk_oneISI_09_query}
		\begin{tabular}{c|ccc}
			\hline
			Decoder & $3$ dB & $5$ dB & $7$ dB \\
			\hline
			ORBGRAND & $1.63\times10^{4}$ & $7.29\times10^{3}$ & $514$ \\
			ORBGRAND-AI ($B=4$) & $4.93\times10^{3}$ & $397$ & $4.46$ \\
			ORBGRAND-AI ($B=8$) & $4.17\times10^{3}$ & $237$ & $2.52$ \\
			ORBGRAND-ISI & $3.58\times10^{3}$ & $125$ & $1.60$ \\
			CDF-ORBGRAND-ISI & $3.53\times10^{3}$ & $98.7$ & $1.60$ \\
			SGRAND-ISI & $3.46\times10^{3}$ & $82.4$ & $1.38$ \\
			\hline
		\end{tabular}
	\end{table}


	
	We next consider a first-order ISI channel with stronger memory
	$h_0=\sqrt{0.6}$ and $h_1=\sqrt{0.4}$, and the corresponding
	results are shown in Fig.~\ref{fig:bpsk_oneISI_06}. As the ISI
	becomes more pronounced, all the considered decoders experience
	performance degradation, with conventional ORBGRAND being affected most severely. The degradation of the proposed GRAND-ISI algorithms is comparatively moderate, making their performance advantages over ORBGRAND-AI and turbo equalization more pronounced. In particular, at a BLER of $10^{-3}$, ORBGRAND-ISI, CDF-ORBGRAND-ISI, and
	SGRAND-ISI provide approximately $0.7$-$1$ dB gains over
	ORBGRAND-AI with $B=8$. Among them, SGRAND-ISI achieves the best
	performance and remains close to the genie-aided ML lower bound.
	
	\begin{figure}[htbp]
		\centering 
		\includegraphics[width=0.46\textwidth]{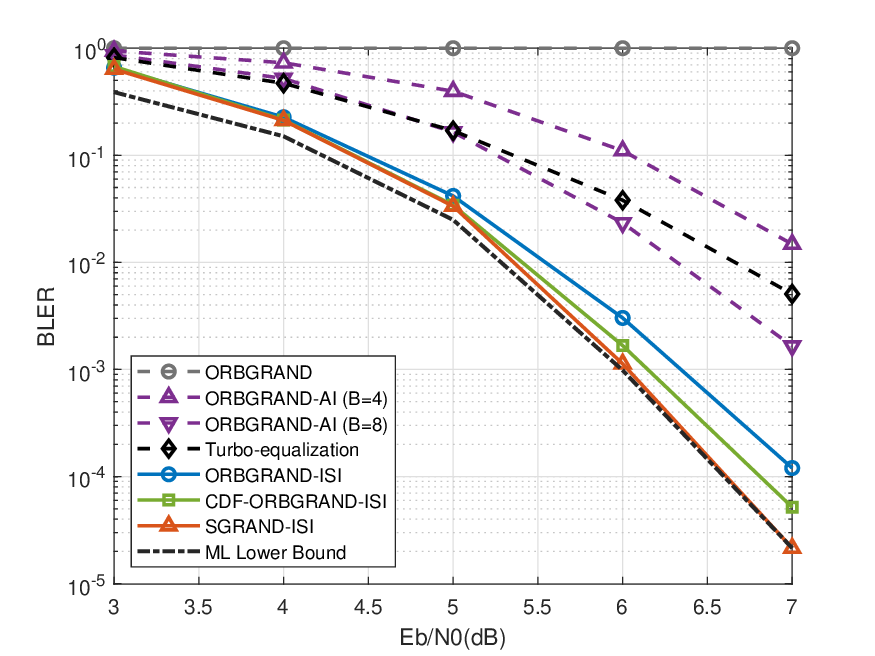}
		\caption{BLER performance with BPSK modulation over a first-order ISI channel with $h_0=\sqrt{0.6}$ and $h_1=\sqrt{0.4}$.}
		\label{fig:bpsk_oneISI_06}
	\end{figure}
	
	\subsubsection{Higher-Order ISI}
	
	We consider a second-order ISI channel and employ a BCH$(127,113)$ code to evaluate the proposed decoding algorithms under higher-order channel memory.
	
	We first investigate the effect of the parameter $g$ for
	R-SGRAND-ISI over the second-order ISI channel. The corresponding
	performance-complexity tradeoff is summarized in
	Table~\ref{tab:g_tradeoff}. Increasing $g$ generally improves the
	decoding performance, since more partially decomposable error bursts
	are retained. However, the improvement gradually diminishes as $g$
	increases. In particular, increasing $g$ from $3$ to $4$ reduces
	the required $E_b/N_0$ at a BLER of $10^{-3}$ by only about
	$0.03$ dB, while increasing the number of retained error bursts from $8624$ to $9480$. We therefore set $g=3$ in
	the subsequent simulations.
	
	\begin{table}[htbp]
		\centering
		\caption{Performance-complexity tradeoff of R-SGRAND-ISI.}
		\label{tab:g_tradeoff}
		\begin{tabular}{c|c|c}
			\hline
			$g$
			& Number of retained error bursts
			& $E_b/N_0$ at BLER $10^{-3}$ (dB)\\
			\hline
			1 & 8128 & 6.00 \\
			2 & 8253 & 5.84 \\
			3 & 8624 & 5.67 \\
			4 & 9480 & 5.64 \\
			\hline
		\end{tabular}
	\end{table}

	The corresponding decoder comparison with $g=3$ is shown in
	Fig.~\ref{fig:bpsk_twoISI_ORBGRAND}. Conventional ORBGRAND suffers
	a substantial performance loss because it ignores the channel memory,
	whereas ORBGRAND-AI significantly improves the performance by
	exploiting the ISI-induced correlation. As also observed for the
	first-order ISI channels, the proposed GRAND-ISI algorithms provide
	further gains. At a BLER of $10^{-3}$, R-ORBGRAND-ISI and
	R-CDF-ORBGRAND-ISI achieve approximately $0.8$ dB and $1$ dB gains,
	respectively, over ORBGRAND-AI with $B=8$. R-SGRAND-ISI achieves
	the best BLER performance and remains close to the genie-aided ML lower
	bound.
	
	\begin{figure}[htbp]
		\centering 
		\includegraphics[width=0.46\textwidth]{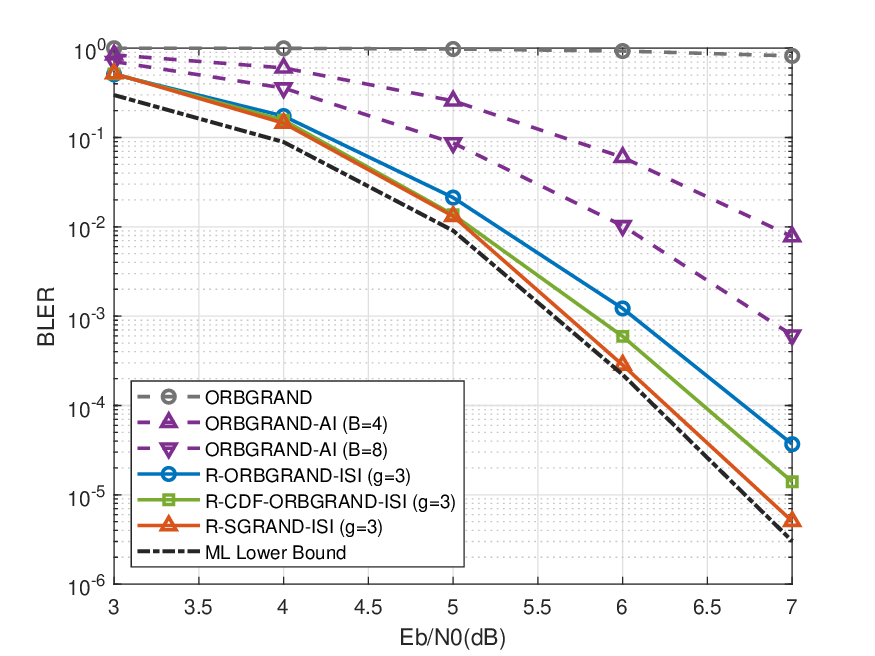}
		\caption{BLER performance with BPSK modulation over a second-order ISI
			channel with $h_0=\sqrt{0.8}$, $h_1=\sqrt{0.15}$, and
			$h_2=\sqrt{0.05}$.}
		\label{fig:bpsk_twoISI_ORBGRAND}
	\end{figure}

	\subsection{Higher-Order Modulation}
	\label{subsec:numerical-higher-order-modulation}
	
	We next consider higher-order modulation to evaluate the extension of
	GRAND-ISI beyond BPSK. Specifically, Gray-labeled 16-QAM signaling is
	considered over the first-order ISI channel with real-valued taps
	$h_0=\sqrt{0.9}$ and $h_1=\sqrt{0.1}$. A CRC-aided polar code of
	length $N=256$ with $251$ nonfrozen bits, including an $11$-bit CRC,
	is employed, corresponding to $240$ information bits and $N_s=64$
	modulation symbols.
	
	\begin{table}[htbp]
		\centering
		\caption{Effect of the pruning parameters for 16-QAM.}
		\label{tab:pruning-selection}
		\begin{tabular}{c|c|c|c}
			\hline
			Fixed & Varied & $V_{T,\ell_{\max}}$ &
			$E_b/N_0$ at BLER $10^{-3}$ (dB) \\
			\hline
			\multirow{4}{*}{$T=4$}
			& $\ell_{\max}=1$ & 256   & 10.79 \\
			& $\ell_{\max}=2$ & 1264  & 9.94  \\
			& $\ell_{\max}=3$ & 5232  & 9.77  \\
			& $\ell_{\max}=4$ & 20848 & 9.65  \\
			\hline
			\multirow{4}{*}{$\ell_{\max}=3$}
			& $T=2$ & 876  & 10.42 \\
			& $T=3$ & 2433 & 9.80  \\
			& $T=4$ & 5232 & 9.77  \\
			& $T=5$ & 9645 & 9.74  \\
			\hline
		\end{tabular}
	\end{table}
	
	The pruning strategy is controlled by $T$, the number of retained
	replacement symbols at each symbol position, and $\ell_{\max}$, the maximum length of the elementary symbol error bursts whose sequence reliabilities are explicitly evaluated. Longer symbol error bursts can still be generated by combining multiple retained elementary symbol error bursts. Table~\ref{tab:pruning-selection} reports the
	BLER performance of P-SGRAND-ISI for different choices of $T$ and
	$\ell_{\max}$. The performance gain becomes limited beyond
	$\ell_{\max}=3$ and $T=4$, while $V_{T,\ell_{\max}}$ continues to
	increase rapidly. We therefore set $T=4$ and $\ell_{\max}=3$ in the subsequent simulations.
	
	\begin{figure}[htbp]
		\centering
		\includegraphics[width=0.46\textwidth]{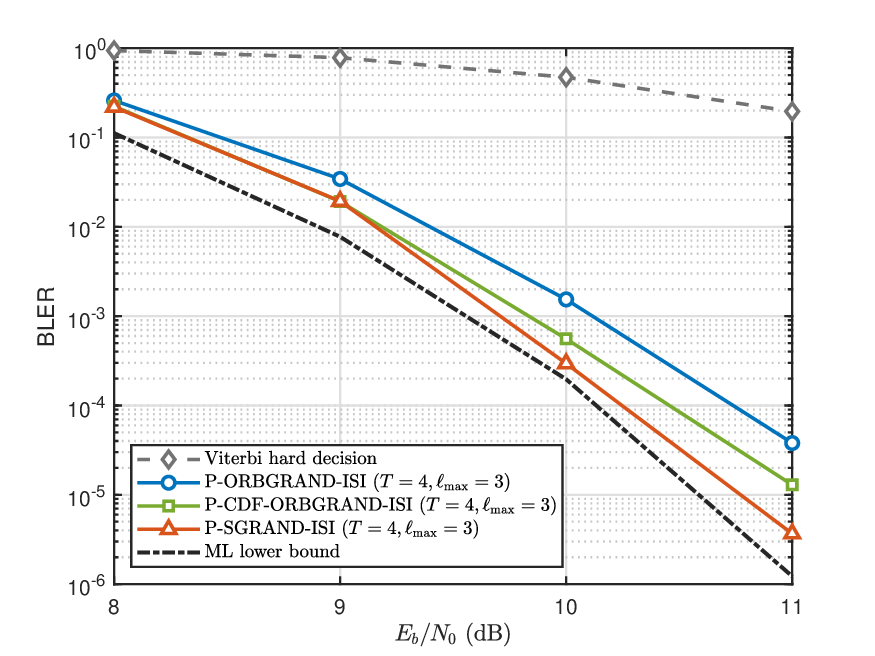}
		\caption{BLER performance of the proposed GRAND-ISI decoders with
			16-QAM modulation for $T=4$ and $\ell_{\max}=3$.}
		\label{fig:16QAM}
	\end{figure}
	
	Fig.~\ref{fig:16QAM} shows the resulting BLER performance.
	P-CDF-ORBGRAND-ISI consistently outperforms P-ORBGRAND-ISI, providing an approximately $0.28$ dB gain at a BLER of $10^{-3}$.
	P-SGRAND-ISI further improves the performance and operates within
	approximately $0.15$ dB of the ML bound. These results indicate that the proposed pruning strategy retains most of the performance benefit of sequence-reliability-based ordering while substantially reducing the number of elementary symbol error bursts that need to be explicitly evaluated. They also demonstrate that the proposed GRAND-ISI framework can be effectively extended from BPSK to higher-order modulation.

	\subsection{Complexity Analysis}
	\label{subsec:complexity-analysis}
	
	We next compare the computational complexity of ORBGRAND-ISI and
	ORBGRAND-AI, with emphasis on their online implementation costs.
	SGRAND-ISI is excluded because its online computation of exact
	sequence reliabilities and dynamic EP generation involve substantially
	different and more dynamic processing, making a direct complexity
	comparison less meaningful. CDF-ORBGRAND-ISI is
	not considered separately because its additional CDF-based processing
	is performed offline, and its per-query online processing is
	essentially the same as that of ORBGRAND-ISI.

	A recent hardware implementation of ORBGRAND-AI has demonstrated
	that reliability sorting and parallel query generation can be efficiently
	realized in hardware \cite{kizilates2025low}. Although
	ORBGRAND-ISI involves different reliability calculations and
	compatibility checks, it retains the rank-based EP generation principle,
	suggesting that similar implementation strategies may be applicable.

	Since the two decoders involve different types of operations, the number of queries alone is insufficient for a fair complexity comparison. We therefore decompose their online complexity into four modules: hard-decision generation, reliability calculation and ordering, error-event compatibility checking and EP generation, and
	parity-check-matrix processing. Here, compatibility checking refers to the validity checking of burst selections for ORBGRAND-ISI and
	block-substitution conflict checking for ORBGRAND-AI.

	For concreteness, we specialize the complexity analysis to BPSK
	signaling over a first-order ISI channel. Let $N$ and $K$ denote the
	code length and dimension, respectively, and let $B_j$ denote the
	length of the $j$-th block, where
	$\sum_{j=1}^{J}B_j=N$. The numbers of reliability values processed
	by ORBGRAND-ISI and ORBGRAND-AI are, respectively,
	\begin{equation}
		V_{\mathrm I}=\frac{N(N+1)}{2},
		\qquad
		V_{\mathrm A}=\sum_{j=1}^{J}(2^{B_j}-1).
	\end{equation}
	Since the considered modules involve different types of operations,
	their complexities are reported in
	Table~\ref{tab:complexity}, with different operation types listed separately.
	
	\begin{table}[htbp]
		\centering
		\caption{Asymptotic online complexity of ORBGRAND-ISI and
			ORBGRAND-AI.}
		\label{tab:complexity}
		\scriptsize
		\setlength{\tabcolsep}{3pt}
		\renewcommand{\arraystretch}{0.9}
		\begin{tabular}{@{}p{0.27\columnwidth}p{0.67\columnwidth}@{}}
			\toprule
			\textbf{Decoder} & \textbf{Complexity} \\
			\midrule
			
			\multicolumn{2}{@{}l}{\textbf{Hard-decision generation}}\\
			ORBGRAND-ISI
			&
			Real: $\mathcal{O}(N)$;
			Comp.: $\mathcal{O}(N)$.
			\\
			ORBGRAND-AI
			&
			Real:
			$\mathcal{O}\!\left(\sum_{j=1}^{J}B_j2^{B_j}\right)$;
			Comp.: $\mathcal{O}(V_{\mathrm A})$.
			\\[0.5mm]
			
			\multicolumn{2}{@{}l}{\textbf{Reliability calculation and ordering}}\\
			ORBGRAND-ISI
			&
			Real: $\mathcal{O}(V_{\mathrm I})$;
			Comp.: $\mathcal{O}(V_{\mathrm I}\log V_{\mathrm I})$.
			\\
			ORBGRAND-AI
			&
			Real: $\mathcal{O}(V_{\mathrm A})$;
			Comp.: $\mathcal{O}(V_{\mathrm A}\log V_{\mathrm A})$.
			\\[0.5mm]
			
			\multicolumn{2}{@{}l}{\textbf{Compatibility checking and EP generation}}\\
			ORBGRAND-ISI
			&
			Comp.:
			$\mathcal{O}\!\left(
			\sum_{q=1}^{Q_{\mathrm{tot,I}}}s_{\mathrm I,q}^{2}
			\right)$;
			Bit: $\mathcal{O}(Q_{\mathrm I}N)$.
			\\
			ORBGRAND-AI
			&
			Comp.:
			$\mathcal{O}\!\left(
			\sum_{q=1}^{Q_{\mathrm{tot,A}}}s_{\mathrm A,q}
			\right)$;
			Bit: $\mathcal{O}(Q_{\mathrm A}N)$.
			\\[0.5mm]
			
			\multicolumn{2}{@{}l}{\textbf{Parity-check-matrix processing}}\\
			ORBGRAND-ISI
			&
			XOR:
			$\mathcal{O}\!\left(Q_{\mathrm I}(N-K)N\right)$.
			\\
			ORBGRAND-AI
			&
			XOR:
			$\mathcal{O}\!\left(Q_{\mathrm A}(N-K)N\right)$.
			\\
			
			\bottomrule
		\end{tabular}
		
		\vspace{1mm}
		\begin{minipage}{0.95\columnwidth}
			\scriptsize
			Real: real-valued operations; Comp.: comparisons; Bit: bit operations.
			Comparison types are not distinguished since only asymptotic
			complexity orders are considered.
		\end{minipage}
	\end{table}
	
	The complexity terms in Table~\ref{tab:complexity} are evaluated as follows.
	For hard-decision generation, ORBGRAND-ISI obtains the hard-decision sequence
	using the Viterbi algorithm, whereas ORBGRAND-AI first partitions the received
	sequence into blocks and enumerates the hard-decision sequences
	within each block. For reliability ordering, standard comparison-based sorting is assumed.
	For compatibility checking, $Q_{\mathrm{tot,I}}$ denotes the total number of
	burst-selection vectors examined by ORBGRAND-ISI, including both valid and
	invalid ones, while $s_{\mathrm I,q}$ denotes the number of error bursts
	selected in the $q$-th examined query among these $Q_{\mathrm{tot,I}}$ queries.
	Similarly, $Q_{\mathrm{tot,A}}$ denotes the total number of selection vectors
	examined by ORBGRAND-AI, and $s_{\mathrm A,q}$ denotes the number of blocks
	selected in the $q$-th examined query among these $Q_{\mathrm{tot,A}}$ queries.
	Pairwise burst-compatibility checking in ORBGRAND-ISI scales as
	$\mathcal{O}(s_{\mathrm I,q}^{2})$, whereas block-conflict checking in
	ORBGRAND-AI scales as $\mathcal{O}(s_{\mathrm A,q})$, assuming that the
	block indices are stored in advance. Moreover, $Q_{\mathrm I}$ and $Q_{\mathrm A}$ denote the numbers of valid
	burst-selection vectors and valid block-selection vectors entering EP
	generation and parity checking for ORBGRAND-ISI and ORBGRAND-AI,
	respectively. Explicit length-$N$ EP construction and direct syndrome
	computation using the full $(N-K)\times N$ parity-check matrix are assumed.

	\begin{table}[htbp]
		\centering
		\caption{Online operation counts per received block for
			BCH$(127,113)$ at $E_b/N_0=5$ dB.}
		\label{tab:complexity-example}
		\scriptsize
		\setlength{\tabcolsep}{3pt}
		\renewcommand{\arraystretch}{0.9}
		
		\begin{tabular}{@{}p{0.31\columnwidth}p{0.63\columnwidth}@{}}
			\toprule
			\textbf{Decoder} & \textbf{Operation count} \\
			\midrule
			
			\multicolumn{2}{@{}l}{%
				\textbf{Hard-decision generation}}\\
			ORBGRAND-ISI
			&
			A/S: \(1010\);
			M: \(506\);
			D: \(506\);
			Cmp-R: \(253\).
			\\
			
			ORBGRAND-AI (\(B=4\))
			&
			A/S: \(6440\);
			M: \(2496\);
			D: \(2496\);
			Exp: \(976\);
			Log: \(488\);
			Cmp-R: \(960\).
			\\
			
			ORBGRAND-AI (\(B=8\))
			&
			A/S: \(81536\);
			M: \(35328\);
			D: \(35328\);
			Exp: \(7424\);
			Log: \(3712\);
			Cmp-R: \(7664\).
			\\
			\addlinespace[0.5mm]
			
			\multicolumn{2}{@{}l}{%
				\textbf{Reliability calculation and ordering}}\\
			ORBGRAND-ISI
			&
			A/S: \(25017\);
			Abs: \(8128\);
			Cmp-B: \(126\);
			Cmp-R: \(\leq 97473\).
			\\
			
			ORBGRAND-AI (\(B=4\))
			&
			A/S: \(472\);
			Abs: \(472\);
			Cmp-R: \(\leq 3737\).
			\\
			
			ORBGRAND-AI (\(B=8\))
			&
			A/S: \(3952\);
			Abs: \(3952\);
			Cmp-R: \(\leq 43329\).
			\\
			\addlinespace[0.5mm]
			
			\multicolumn{2}{@{}l}{%
				\textbf{Compatibility checking and EP generation}}\\
			ORBGRAND-ISI
			&
			Cmp-I: \(1938.8\);
			A-int: \(969.4\);
			W: \(15847.7\);
			U: \(524.9\);
			XOR: \(15847.7\).
			\\
			
			ORBGRAND-AI (\(B=4\))
			&
			Cmp-I: \(2232.4\);
			W: \(51047.6\);
			U: \(1764.5\);
			XOR: \(51047.6\).
			\\
			
			ORBGRAND-AI (\(B=8\))
			&
			Cmp-I: \(2012.9\);
			W: \(30124\);
			U: \(1033.8\);
			XOR: \(30124\).
			\\
			\addlinespace[0.5mm]
			
			\multicolumn{2}{@{}l}{%
				\textbf{Parity-check-matrix processing}}\\
			ORBGRAND-ISI
			&
			XOR: \(220120\).
			\\
			
			ORBGRAND-AI (\(B=4\))
			&
			XOR: \(709039\).
			\\
			
			ORBGRAND-AI (\(B=8\))
			&
			XOR: \(418415\).
			\\
			
			\bottomrule
		\end{tabular}
		
		\vspace{1mm}
		\begin{minipage}{0.96\columnwidth}
			\scriptsize
			A/S: real-valued additions/subtractions;
			A-int: integer additions;
			M: multiplications;
			D: divisions;
			Exp: exponentials;
			Log: logarithms;
			Abs: absolute-value operations;
			Cmp-R: real-valued comparisons;
			Cmp-B: binary-value comparisons;
			Cmp-I: integer-index comparisons;
			W: initialization writes;
			U: support-position updates.
			Data-dependent counts are averaged over received blocks, whereas
			the Cmp-R counts for reliability ordering denote upper bounds.
		\end{minipage}
	\end{table}

	The reported ORBGRAND-AI complexity is based on our implementation
	and may differ from the implementation choices in
	\cite{duffy2023using,duffy2025decoding,kizilates2025low}. For BCH$(127,113)$ at
	$E_b/N_0=5$ dB over $10^6$ frames, the average total/valid query counts are $(172.4,124.8)$ for
	ORBGRAND-ISI and $(578.3,402.0)$ and $(558.4,237.2)$ for
	ORBGRAND-AI with $B=4$ and $B=8$, respectively. The corresponding operation counts are summarized in
	Table~\ref{tab:complexity-example}, where the ordering counts are
	worst-case bounds for merge sort and the query-dependent counts are
	Monte Carlo averages. Each valid query recomputes the syndrome
	using the full $14\times127$ parity-check matrix, requiring
	$14\times126$ XOR operations.

	The results show a clear module-wise tradeoff. ORBGRAND-ISI requires
	fewer operations for hard-decision generation and, owing to its
	smaller valid query count, also incurs fewer operations for EP
	generation and parity-check-matrix processing. Its larger reliability-processing operation counts arise from the
	evaluation and ordering of more reliability values; nevertheless, it is
	exactly this finer reliability processing that leads to more effective
	EP ordering and consequently fewer valid queries. The number of operations required
	for compatibility checking remains comparable to that of
	ORBGRAND-AI. Overall, the two decoders exhibit different complexity
	profiles across the individual processing modules.

	\section{Conclusion}\label{Conclusion}
	
	In this paper, we extended GRAND to linear Gaussian ISI channels by
	exploiting channel memory through error bursts and sequence
	reliability. Based on these concepts, SGRAND-ISI, ORBGRAND-ISI, and
	CDF-ORBGRAND-ISI were developed, where SGRAND-ISI is equivalent to
	ML decoding under complete querying. The framework was further
	extended to higher-order ISI channels and higher-order modulation.
	
	Numerical results for BCH and polar codes show that the proposed
	algorithms substantially outperform GRAND variants that ignore the
	ISI structure, while CDF-ORBGRAND-ISI closely approaches the ML
	benchmark. The higher-order modulation results further demonstrate the
	effectiveness of the proposed GRAND-ISI algorithms.
	The corresponding online decoding complexity was also evaluated to
	characterize the computational costs of the individual processing modules.
	
	\section{Appendices}
	
	\subsection{Proof of Lemma~\ref{lemma:1}}
	\label{Lemma1}

	Under the first-order ISI model, the $i$-th summand in
	\eqref{eq:sequence-metric} depends only on $x_{i-1}$ and $x_i$.
	Hence, flipping the symbol at position $s$ can affect only the $s$-th
	and $(s+1)$-th summands, with the latter absent when $s=N$.
	
	For $j\in\{1,2\}$, let
	$\mathcal{A}_j\triangleq
	\bigcup_{s\in\mathcal{S}_j}[s:\min\{s+1,N\}]$
	denote the affected summand indices. The separation condition in
	Lemma~\ref{lemma:1} implies
	$\mathcal{A}_1\cap\mathcal{A}_2=\emptyset$. Therefore, the flips in
	$\mathcal{S}_1$ and $\mathcal{S}_2$ affect disjoint summands of the
	sequence metric, and Definition~\ref{defn:sequence-reliability} gives
	\begin{align}
		\operatorname{Rel}(\mathcal{S}_1\cup\mathcal{S}_2)
		&=
		\sum_{j=1}^{2}
		\sum_{i\in\mathcal{A}_j}
		\ln
		\frac{
			p_{\mathsf{Y}_i|\underline{\mathsf{X}}_{i-1:i}}
			\bigl(y_i|\underline{x}^{*}_{i-1:i}\bigr)
		}{
			p_{\mathsf{Y}_i|\underline{\mathsf{X}}_{i-1:i}}
			\left(
			y_i\Big|
			\left[
			\mathcal{F}_{\mathcal{S}_j}(\underline{x}^{*})
			\right]_{i-1:i}
			\right)
		}
		\nonumber\\
		&=
		\operatorname{Rel}(\mathcal{S}_1)
		+
		\operatorname{Rel}(\mathcal{S}_2),
	\end{align}
	which completes the proof.
	
	\subsection{Proof of Lemma~\ref{lemma:2}}
	\label{Lemma2}
	
	The result is immediate for $a=b$. For $a<b$, let
	$\mathcal{I}_k\triangleq[a:k]$. Expanding
	Definition~\ref{defn:sequence-reliability} and substituting the Gaussian
	ISI likelihood gives, for $k=a+1,\ldots,b$,
	\begin{align}
		\operatorname{Rel}(\mathcal{I}_k)
		-
		\operatorname{Rel}(\mathcal{I}_{k-1})
		=
		\operatorname{Rel}(\{k\})
		+
		\frac{4h_0h_1}{\sigma^2}
		x_{k-1}^*x_k^*.
	\end{align}
	Summing over $k=a+1,\ldots,b$ and using
	$\mathcal{I}_a=\{a\}$ completes the proof.

	\subsection{Proof of Theorem~\ref{thm:SGRAND-ISI-ML}}
	\label{Equivalence}

	For SGRAND-ISI, where
	$\gamma(\underline{y},\mathcal{S})
	=\operatorname{Rel}(\mathcal{S})$, let the decoder output be
	$\widehat{w}_{\gamma,\mathrm{SGRAND\text{-}ISI}}$.
	Using \eqref{eq:aggregate-reliability},
	\eqref{eq:error-burst-additivity},
	Definition~\ref{defn:sequence-reliability}, and
	\eqref{eq:EP-flipping-relation}, we obtain, for every message $w$,
	\begin{align}
		\rho_{\gamma}
		\bigl(\underline{y},\underline{e}(w)\bigr)
		&=
		\operatorname{Rel}\!\left(
		\operatorname{supp}(\underline{e}(w))
		\right)
		\nonumber\\
		&=
		\Gamma(\underline{x}^{*},\underline{y})
		-
		\Gamma(\underline{x}(w),\underline{y}),
		\label{eq:sgrand-isi-codeword-metric}
	\end{align}
	with the same equality holding for
	$\underline{e}(w)=\underline{0}$ under the empty-sum convention.
	Therefore,
	\begin{align}
		\widehat{w}_{\gamma,\mathrm{SGRAND\text{-}ISI}}
		&=
		\arg\min_{w=1,\ldots,M}
		\rho_{\gamma}
		\bigl(\underline{y},\underline{e}(w)\bigr)
		\nonumber\\
		&=
		\arg\max_{w=1,\ldots,M}
		\Gamma(\underline{x}(w),\underline{y})
		=
		\widehat{w}_{\mathrm{ML}},
	\end{align}
	which completes the proof.
	
	\subsection{Proof of Lemma~\ref{lemma:higher-order-additivity}}
	\label{proof:higher-order-additivity}
	
	
	Under the order-$L$ ISI model, flipping the symbol at position $s$ can affect only the summands in \eqref{eq:sequence-metric} indexed by $s,\ldots,\min\{s+L,N\}$. The separation condition
	$\min(\mathcal{S}_2)-\max(\mathcal{S}_1)>L$
	therefore ensures that the flips in $\mathcal{S}_1$ and $\mathcal{S}_2$ affect disjoint summands of the sequence metric. Applying Definition~\ref{defn:sequence-reliability} then completes the proof.

	\subsection{Proof of Lemma~\ref{lemma:mary-additivity}}
	\label{app:proof-mary-additivity}
	
%
		
	Under the first-order ISI model, replacing the symbol at position $r$
	can affect only the $r$-th and $(r+1)$-th summands in
	\eqref{eq:mary-sequence-metric}, with the latter absent when
	$r=N_s$. The separation condition
	$\max(\mathcal{D}_1)+1<\min(\mathcal{D}_2)$ therefore ensures that the
	replacements specified by $\mathcal{Q}_1$ and $\mathcal{Q}_2$ affect
	disjoint summands of the sequence metric. Hence,
	\eqref{eq:mary-additivity} follows directly, completing the proof.

	\bibliographystyle{IEEEtran}
	\bibliography{ref.bib}
	
\end{document}